\documentclass{LMCS}

\def\dOi{9(3:25)2013}
\lmcsheading%
{\dOi}
{1--46}
{}
{}
{Jan.~16, 2013}
{Sep.~24, 2013}
{}

\usepackage{amsmath,amssymb,graphicx,xspace,qtree,color}
\usepackage[utf8]{inputenc}
\usepackage{hyperref}

\newcounter{tbsnr}
\newenvironment{tbs}
{\addtocounter{tbsnr}{1}\par\bigskip \noindent\fbox{\thetbsnr}
\hspace{2mm}\begin{minipage}{.9\linewidth}\tt \small}
{\end{minipage}\hspace*{\fill}\bigskip}

\newtheorem{claim}[thm]{Claim}

\newcommand\mucalc{\ensuremath{\mu\text{-calculus}}\xspace}
\newcommand\ml{\ensuremath{\textup{ML}}\xspace}
\newcommand\fo{\ensuremath{\textup{FO}}\xspace}
\newcommand\mso{\ensuremath{\textup{MSO}}\xspace}
\newcommand\fp{\ensuremath{\textup{LFP}}\xspace}
\newcommand\gfo{\ensuremath{\textup{GFO}}\xspace}
\newcommand\gfp{\ensuremath{\textup{GFP}}\xspace}
\newcommand\unfo{\textup{UNFO}\xspace}
\newcommand\unfp{\textup{UNFP}\xspace}
\newcommand\ctlx{\ensuremath{\textup{CTL}^*(\textrm{\sf X})}\xspace}
\newcommand{\unfolet}{\text{\rm$\unfo(\mathsf{let})$}\xspace} 
\newcommand{\sunfp}{\textup{sUNFP}\xspace}
\newcommand{\sunfo}{\textup{sUNFO}\xspace}

\newcommand\inc{\textit{inc}}
\newcommand{\subf}{\textsc{subf}\xspace}
\newcommand{\types}{\textsc{ntypes}\xspace}
\newcommand\jump{\textup{\bf S}}

\newcommand\ptime{\textup{PTime}\xspace}
\newcommand\exptime{\textup{ExpTime}\xspace}
\newcommand\twoexptime{2{\textup{ExpTime}}\xspace}

\newcommand\np{\textup{NP}\xspace}
\newcommand\conp{\textup{coNP}\xspace}
\newcommand\pnplog{\ensuremath{\textup{P}^{\textup{NP}[O(\log n)]}}\xspace}
\newcommand\pnppar{\ensuremath{\textup{P}^{\textup{NP}}_{||}\xspace}}
\newcommand\pnplogsq{\ensuremath{\textup{P}^{\textup{NP}[O(\log^2 n)]}}\xspace}
\newcommand\pnplogi{\ensuremath{\textup{P}^{\textup{NP}[O(\log^i n)]}}\xspace}
\newcommand\pnp{\ensuremath{\textup{P}^{\textup{NP}}}\xspace}
\newcommand{\npnp}{\ensuremath{\textup{NP}^{\textup{NP}}}\xspace}
\newcommand{\conpnp}{\ensuremath{\textup{coNP}^{\textup{NP}}}\xspace}
\newcommand{\lexsat}{\textup{LEX(SAT)}\xspace}
\newcommand{\lexisat}{\ensuremath{\textup{LEX}_i\textup{(SAT)}}\xspace}
\newcommand\set[1]{\ensuremath{\{#1\}}\xspace}
\newcommand\tuple[1]{\ensuremath{\overline{#1}}\xspace}

\newcommand{\bisim}{\ensuremath{\approx}}
\newcommand{\unbisim}{\ensuremath{\approx_{\text{UN}}}}
\newcommand{\unkbisim}{\ensuremath{\approx_{\text{UN}^k}}}

\newcommand{\restrict}{\ensuremath{|}\xspace}
\newcommand\dom{\text{dom}}

\newtheorem{theorem}{Theorem}[section]
\newtheorem{example}[theorem]{Example}
\newtheorem{lemma}[theorem]{Lemma}

\newtheorem{definition}[theorem]{Definition}
\newtheorem{remark}[theorem]{Remark}

\begin{document}
\title[Unary Negation]{Unary negation}

\author[B.~ten Cate]{Balder ten Cate\rsuper a}	
\address{{\lsuper a}UC Santa Cruz}	
\email{btencate@ucsc.edu}  
\thanks{{\lsuper a}Balder ten Cate has been funded partially by the ERC grant Webdam, agreement 226513, and partially
by the NSF grants IIS-0905276 and IIS-1217869.}

\author[L.~Segoufin]{Luc Segoufin\rsuper b}	
\address{{\lsuper b}INRIA and ENS Cachan, LSV}	
\email{luc.segoufin@inria.fr}  



\ACMCCS{[{\bf Theory of computation}]: Logic; Formal languages and automata theory; Theory and algorithms for application domains---Database theory}

\begin{abstract}
  We study fragments of first-order logic and of least fixed point
  logic that allow only unary negation: negation of formulas with at
  most one free variable. These logics generalize many interesting
  known formalisms, including modal logic and the $\mu$-calculus, as well
  as conjunctive queries and monadic Datalog. We show that
  satisfiability and finite satisfiability are decidable for both
  fragments, and we pinpoint the complexity of satisfiability, finite
  satisfiability, and model checking. We also show that the
  unary negation fragment of first-order logic is model-theoretically
  very well behaved. In particular, it enjoys Craig Interpolation and
  the Projective Beth Property.
\end{abstract}
\maketitle

\section{Introduction}

Vardi~\cite{Vardi96} raised the question ``why is modal logic so robustly
decidable?''.  His explanation centers around the fact that modal
logic has the tree-model property. More precisely, modal logic enjoys a combination of
three properties, namely (i) the \emph{tree-model property} (if a sentence has
a model, it has a model that is a tree), (ii) \emph{translatability into tree
  automata} (each formula can be transformed into a tree automaton, or
equivalently, an MSO formula, recognizing
its tree models), and (iii) the \emph{finite
  model property} (if a formula  has a model, it also has a finite model).  These
three properties form a powerful explanation for the decidability of the
satisfiability problem, and the finite satisfiability problem, for modal logic and many of its extensions such
as the modal \mucalc.  The guarded fragment of first-order logic (\gfo) was
proposed by Andr{\'e}ka, van Benthem and N{\'e}meti \cite{andr:moda98} as a large
fragment of first-order logic that generalizes modal logic while essentially retaining
these properties. It consists of \fo formulas in which all quantifiers are
``guarded'' by atomic formulas. \gfo has the tree-like model property (if a
sentence has a model, it has a model of bounded tree width), it can be
translated into tree automata (each formula can be transformed into a tree automaton
recognizing the tree decompositions of its models of bounded tree width) and it
has the finite model property \cite{andr:moda98,Gradel99}.

In this paper we provide another, orthogonal generalization of modal logic that
enjoys the same nice properties. We introduce \unfo, a fragment of \fo in which
\emph{negation} is restricted to formulas having only one free variable.  \unfo
is incomparable in terms of expressive power to \gfo. It generalizes modal
logic, as well as other formalisms, such as conjunctive queries, that are not
contained in \gfo. We show that \unfo has the tree-like model property, is
translatable into tree-like automata (in the sense described above), and has the
finite model property. Hence \unfo, too, is robustly decidable.

We also introduce \unfp, which extends \unfo with least and greatest monadic
fixpoints, in the same way that the \mucalc extends modal logic \cite{Kozen1983}, and 
guarded fixpoint logic (\gfp) extends \gfo \cite{GradelW99}. \unfp generalizes the 
\mucalc but also monadic Datalog and remains incomparable with
\gfp. It still has the tree-like model property and can be translated into
\mso, but it no longer has the finite model
property. Nevertheless, we show that finite satisfiability for \unfp
is decidable (note that the decidability of the analogous problem for \gfp was
only recently solved in~\cite{Mikovince}).
More precisely, the satisfiability problem is \twoexptime-complete, both for \unfo and for
\unfp, both on arbitrary and finite structures.

We also study the model checking problem. In contrast with \gfo, whose model
checking problem is \ptime-complete~\cite{BerwangerG01}, we show
that for \unfo it is complete for \pnplogsq, providing
one of the few natural complete problems for that complexity class.  For \unfp,
model checking is hard for \pnp and contained in
$\npnp\cap\conpnp$. The gap between the upper bound and the lower bound
reflects a similar open problem for \gfp and the \mucalc
where the model checking problem lies between \ptime and $\np\cap\conp$~\cite{BerwangerG01}.

\unfo is not only computationally but also model-theoretically very
well behaved. 
We characterize the expressive power of \unfo in terms of an
appropriate notion of invariance, and we show that \unfo has 
Craig Interpolation as well as the Projective Beth
Property. Note that Craig Interpolation fails for
\gfo~\cite{HooglandM02}.
On trees, \unfo and \unfp correspond to well-known existing
formalisms. 

\medskip

\textbf{Outline of the paper.}  In Section~\ref{section-prelim}, we formally
introduce \unfo and \unfp, and we review relevant background material on modal
logics and computational complexity.  In Section~\ref{section-express} we
develop the model theory of \unfo and \unfp: we introduce an appropriate notion
of bisimulations, we state a finite model property and a tree-like model property,
we obtain model theoretic characterizations, and we prove Craig
Interpolation and the Projective Beth Property for \unfo. In
Section~\ref{section-sat}, we show that the satisfiability problem for \unfo,
and for \unfp, is 2\exptime-complete, both on arbitrary structures and on
finite structures. In Section~\ref{section-model-check}, we map out the
complexity of the model checking problem, that is, the problem of evaluating a
formula in a given finite structure. In Section~\ref{section-tree}, we study
the expressive power of \unfo and \unfp on tree structures.  We conclude in
Section~\ref{section-discussion} with a comparison with other work, in
particular on
guarded negation logics.

This paper is the journal version of~\cite{CS11}. It contains new and simpler
proofs for many of the results as well as new results, such as the characterization
theorem in the finite case, cf. Theorem~\ref{thm-charac-fo-finite}.

\section{Preliminaries}\label{section-prelim}
We consider relational structures. A relational schema
is a finite set of relation symbols fixing an arity to each
relation. A \emph{model}, or \emph{structure}, $M$ over a relational schema
$\sigma$ is a set $\dom(M)$, the \emph{domain} of $M$, together with an
interpretation $R^M$ to each
relation symbol $R$ of $\sigma$ as a relation over the domain of the arity given
by the schema.  A model is said to be \emph{finite} if its domain is finite.
We assume familiarity with first-order logic, \fo, and least fixpoint logic, \fp, over
relational structures. We use classical syntax and semantics for \fo and
\fp. In particular we write $M\models \phi(\tuple{u})$ or $(M,\tuple{u}) \models
\phi(\tuple{x})$
for the fact that the
tuple $\tuple{u}$ of elements of the model $M$ makes the \fo-formula, or LFP-formula, $\phi(\tuple{x})$
true on $M$. 

Given a structure $M$ and a set $X\subseteq \dom(M)$ we denote by $M\restrict
X$ the substructure of $M$ induced by $X$.

\subsection{\unfo and \unfp}
\newcommand{\unnormalform}{UN-normal form\xspace}
\newcommand{\width}{width\xspace} \newcommand{\LFP}{\textup{LFP}}
\newcommand{\GFP}{\textup{GFP}} We define the \emph{unary-negation
  fragment of first-order logic} (or \unfo, for short), as the fragment of \fo given
by the following grammar (where $R$ is an arbitrary relation name from
the underlying schema):
\begin{align*}
  \phi::=~~~ & R(\tuple{x}) ~|~ x=y ~|~ \phi \wedge \phi ~|~ \phi \vee \phi ~|~
  \exists x \phi ~|~ \neg \phi(x)
\end{align*}
where, in the last clause, $\phi$ has no free variables besides (possibly)
$x$. Throughout this paper, we will keep using the notation $\phi(x)$ to
indicate that a formula has at most one free variable.
In other words, \unfo is the restriction of \fo where negation is only allowed if
the subformula has at most one free variable. In particular $x \neq y$
is not expressible
in \unfo.

We say that a formula of \unfo is in \unnormalform if, in the syntax tree of
the formula, every existential quantifier (except for the root of the syntax
tree) is either directly below another existential quantification, or the
subformula starting with that quantifier has at most one free variable.
In other words, formulas in \unnormalform
are existential positive formulas in prenex normal form where each atom is
either a positive atom over the underlying schema or a possibly negated formula
with at most one free variable in \unnormalform.

For instance the formula $\exists x \exists y (R(x,y) \land \exists z S(x,y,z)) $
is not in \unnormalform. However the equivalent formula $\exists x \exists y
\exists z ~~ (R(x,y) \land S(x,y,z))$ is in \unnormalform.  Similarly the formula
$$\exists x (R(x) \land \exists y (R(y) \land \exists z (R(z) \land (\exists x
S(x,y,z)))))$$
is not in \unnormalform, but each of the equivalent formulas
\[\exists x \exists y \exists z \exists x' R(x) \land R(y) \land R(z) \land
S(x',y,z)\]
and
\[\exists x (R(x) \land \exists y \exists z \exists x (R(y) \land R(z) \land
S(x,y,z))\]
is.

Every formula of \unfo can be transformed into an equivalent formula in
\unnormalform in linear time by ``pulling out existential quantifiers'' as soon
as the corresponding subformula has more than one free variable, using the
following two rewrite rules:
\[ \phi\land\exists x\psi ~\equiv~ \exists x (\phi\land\psi) ~~\text{
  provided that $x$ does not occur free in $\phi$}\]
\[ \phi\lor\exists x\psi ~\equiv~ \exists x (\phi\lor\psi) ~~\text{
  provided that
  $x$ does not occur free in $\phi$}\]
(together with safe renaming of variables where needed). For instance, starting with
$$\exists x (R(x) \land \exists y (R(y) \land \exists z (R(z) \land (\exists x
S(x,y,z)))))$$
one could obtain:
\begin{align*}
\exists x (R(x) \land \exists y \exists z \exists x (R(y) \land R(z) \land
S(x,y,z)))
\end{align*}

Bringing a \unfo
formula into \unnormalform may increase the number of variables
occurring in the formula, because applying the above rewrite rules may require renaming bound
variables.   A formula of \unfo is said to be \emph{of
  \width~$k$} if it can be put in \unnormalform using the above rules in such a
way that the resulting formula uses at most $k$
variables. The width of the above formula is therefore~3. We denote by
$\unfo^k$ the set of all \unfo formulas of \width~$k$.

In order to define \emph{unary-negation fixpoint logic} (\unfp) we introduce extra unary predicates that will
serve for computing unary fixpoints. We denote the unary predicates given by
the relational schema using the letters $P,Q \ldots$ and the unary predicates
serving for computing the fixpoints by $X,Y \ldots$.
By $\unfo(\tuple{X})$ we mean \unfo defined over the schema extended with the
unary predicates $\tuple{X}$. In particular it allows formulas of the form
$\neg\phi(x,\tuple{X})$. \unfp is the extension of $\unfo(\tuple{X})$ by means
of the following least fixpoint construction:
$$[\LFP_{X,x}~\phi(X,\tuple{X},x)](y)$$
where $X$ occurs only positively in $\phi$. 
An analogous greatest fixed point operator is definable by
dualization. Note that no first-order parameters (i.e., free variables
in the body of $\phi$ other than $x$) are permitted.

Note that \unfp is a syntactic fragment of least fixpoint logic (LFP),
i.e., the extension of full first-order logic with the least fixpoint
operator. Therefore, we can simply refer to the literature on LFP for
the semantics of these formulas (cf.~for
example~\cite{libk:elem04}). However, we will discuss the semantics of
the least fixpoint operator here in some detail, because our arguments
later on will refer to it. Consider any \unfp formula of the
form $$[\LFP_{X,x}~\phi(X,\tuple{X},x)](y)$$ and any structure
$(M,\tuple{S})$, where $\tuple{S}$ is a collection of subsets of the
domain of $M$ that form the interpretation for $\tuple{X}$. Since $X$
occurs in $\phi$ only positively, $\phi(X,\tuple{X},x)](y)$ induces a
monotone operation $\mathcal{O}_\phi$ on subsets of the domain of $M$,
where $\mathcal{O}_\phi(A) = \{a\in \dom(M)\mid (M,\vec{S},
A)\models\phi(a)\}$. By the Knaster-Tarski fixpoint theorem, this
monotone operation has a unique least-fixpoint. By definition, an
element $b\in \dom(M)$ satisfies the formula
$[\LFP_{X,x}~\phi(X,\tuple{X},x)](y)$ in $(M,\vec{S})$ if and only if $b$ belongs
to this least fixpoint. The least fixpoint of the monotone operation
$\mathcal{O}_\phi$ is known to be the
intersection of all its pre-fixed points, i.e., $\bigcap\{A\subseteq
\dom(M) \mid A\supseteq \mathcal{O}_\phi(A)\}$, and it can be
equivalently characterized as
${\mathcal{O}_\phi}^\kappa(\emptyset)$, where $\kappa=|\dom(M)|$,
${\mathcal{O}_\phi}^0(\emptyset) = \emptyset$; for all successor ordinals $\lambda+1$,
${\mathcal{O}_\phi}^{\lambda+1}(\emptyset) =
{\mathcal{O}_\phi}({\mathcal{O}_\phi}^\lambda(\emptyset))$;
and for all limit ordinals $\lambda\leq \kappa$, ${\mathcal{O}_\phi}^{\lambda}(\emptyset) =
\bigcup_{\lambda'<\lambda}{\mathcal{O}_\phi}^{\lambda'}(\emptyset)$.

The same definition of the \unnormalform applies to \unfp. 
As in the case of \unfo, we say that a \unfp formula \emph{has \width $k$} if,
when put in \unnormalform, it uses at most $k$ first-order variables. In other
words, a formula of \unfp has \width $k$ if all the
``$\unfo(\tuple{X})$-parts'' of its subformulas have \width~$k$.
We denote by $\unfp^k$ the set of all \unfp formulas of width~$k$. 

The \emph{negation depth} of a \unfo or \unfp formula will also be an
important parameter. It is the maximal nesting depth of negations in its
syntax tree.

\begin{example}\label{example-prelim}
  Two examples of \unfo formulas are $\exists yzu (R(x,y)\land R(y,z)\land
  R(z,u)\land R(u,x))$, which expresses the fact that $x$ lies on a directed
  $R$-cycle of length 4, and its negation $\neg\exists yzu (R(x,y)\land
  R(y,z)\land R(z,u)\land R(u,x))$. It follows from known results \cite{andr:moda98}
  that neither can be expressed in the guarded fragment, and
  therefore, these examples show that \unfo can express properties
  that are not definable in the guarded fragment. On the other hand,
  we will see in Section~\ref{sec:unbisimulations} that 
  the guarded-fragment formula $\forall xy (R(x,y)\to S(x,y))$ has no
  equivalent in \unfo, and therefore, the two logics are incomparable 
  in expressive power.

  A \emph{conjunctive query} (CQ) is a query defined by a first-order formula of
  the form $\exists x_1\cdots x_n~ \tau_1 \wedge \cdots \wedge \tau_l$, where
  each $\tau_i$ is a (positive) atomic formula. A \emph{union of conjunctive
    queries} (UCQ) is a query defined by a finite disjunction of first-order
  formulas of the above form. Clearly, every UCQ is definable in \unfo. In
  fact, \unfo can naturally be viewed as the extension of the language of UCQs
  with unary negation.  It is also worth noting that, in a similar way, all
  \emph{monadic datalog} queries (i.e., datalog queries in which all IDB
  relations are unary~\cite{Cosmadakis88:decidable}) are definable in \unfp.
  It was shown in~\cite{Cosmadakis88:decidable} that query containment is
  decidable in 2ExpTime for monadic datalog. As the containment of two unary
  Datalog programs can be expressed in \unfp, the decidability of \unfp in
  2ExpTime, cf. Theorem~\ref{theorem-sat-upper-bound}, generalizes this result.
  We also mention that the query containment problem was recently shown to be
  hard for 2Exptime~\cite{BenediktBS12}, hence the lower bound of
  Theorem~\ref{theorem-sat-upper-bound} also follows from this fact.
\end{example}

\subsection{Modal logic and bisimulation}\label{sec-modal-logic}
\unfo and \unfp can be viewed as extensions of modal
logic and the \mucalc, and, actually, of their \emph{global two-way} extensions.  As
several of our proofs will make reductions to the modal logic case, we now
review relevant definitions and results regarding \emph{global two-way modal
  logic} and the \emph{global two-way \mucalc{}}.

We view a \emph{Kripke structure} as a relational structure over a
schema consisting of unary and binary relations. Modal logics are
languages for describing properties of nodes in Kripke structures.
Intuitively, modal formulas can navigate Kripke structures by 
traversing edges  in a forward or backward direction.

We will use \ml to denote the modal language with forward and backward
modalities, and with the global modality, as defined by the following grammar.
$$ \phi ::= P ~|~ \phi \land \phi ~|~ \lnot \phi ~|~ \langle R
\rangle \phi ~|~ \langle R^{-1} \rangle \phi ~|~ \jump\phi$$ where $P$ is a
unary relation symbol (also called \emph{proposition letter} in this setting),
and $R$ is a binary relation symbol (also called an \emph{accessibility
  relation} in this context).  Disjunction and the ``box operators'' $[R]$ and
$[R^{-1}]$ are definable as duals of conjunction, $\langle R\rangle$ and
$\langle R^{-1}\rangle$, respectively (for instance $[R]\phi$ is $\lnot\langle R\rangle\lnot\phi$).  

The semantics of \ml can be given via a translation into \unfo:
For each \ml formula $\phi$ we construct by induction, as explained in Figure~\ref{fig:modal-semantics}, a \unfo formula $\phi^*(x)$ such that for
each Kripke structure $M$ and node $a$ we have $M \models \phi^*(a)$ iff $a$ has the
property $\phi$ on $M$. 

\begin{figure}
\[\begin{array}{lll}
(P)^*(x)&  = & P(x) \\
(\phi\land\psi)^*(x) & = & \phi^*(x) \land \psi^*(x)\\
(\neg\phi)^*(x)& = & \lnot \phi^*(x) \\
(\langle R\rangle\phi)^*(x) & = & \exists y~ R(x,y) \land \phi^*(y) \\
(\langle R^{-1}\rangle\phi)^*(x) & = & \exists y~ R(y,x) \land \phi^*(y) \\
(\jump\phi)^*(x) & = & \exists y~ \phi^*(y)\\
\end{array}\]
\caption{Inductive translation of an \ml-formula $\phi$ to
  an equivalent \unfo-formula $\phi^*(x)$}\label{fig:modal-semantics}
\end{figure}
We
refer to \ml as \emph{global two-way modal logic}, because it
includes the global modal operator $\jump$ and the inverse modal
operators $\langle R^{-1}\rangle$ (and their duals).
Traditionally, the basic modal logic is defined without those features
and can only navigate by traversing an edge in the forward direction.

The \emph{global two-way \mucalc}, which we denote by $\ml_\mu$, is obtained by
adding fixpoint variables and a least fixpoint operator to the language of \ml:
fixpoint variables are admitted as atomic formulas, and whenever $\phi$ is a
formula of $\ml_\mu$ in which a fixpoint variable $X$ occurs only positively
(under an even number of negations), then $\mu X \phi$ is again a valid formula
of $\ml_\mu$, and it denotes the fixpoint of the monotone operation on sets
defined by $\phi(X)$. An analogous greatest fixpoint operator is definable as
the dual of the least fixpoint operator. Adding the rule
\[\begin{array}{lll}
(\mu X \phi)^*(x) & =& [\LFP_{X,y}~\phi^*(y)](x)
\end{array}\]
to the table of Figure~\ref{fig:modal-semantics} shows that $\ml_\mu$ can be seen as a
fragment of \unfp.

We know from~\cite{FisherLadner,vard:reas98} that the satisfiability problem
for $\ml$ and for $\ml_\mu$, on arbitrary Kripke structures, is \exptime-complete. Although
$\ml$ has the finite model property \cite{FisherLadner}, that is, every
satisfiable $\ml$-formula is satisfied in some finite Kripke
structure, the same does not hold for $\ml_\mu$, and
therefore the satisfiability problem for $\ml_\mu$ on finite
structures is not the same problem as the satisfiability problem for
$\ml_\mu$ on arbitrary structures. Nevertheless,
it was shown in~\cite{Mikolaj} that the satisfiability problem for
$\ml_\mu$ on finite Kripke structures  is \exptime-complete.

\begin{theorem}\cite{vard:reas98,Mikolaj}\label{thm-decid-mu}
Testing whether a formula of $\ml_\mu$ is satisfiable is
\exptime-complete, both on arbitrary Kripke structures and on 
finite structures.
\end{theorem}

We note that, while the two-way modal $\mu$-calculus as defined
in~\cite{vard:reas98,Mikolaj} does not include the global modality \jump, the
results from~\cite{Mikolaj} immediately extend to full $\ml_\mu$.

Modal formulas are invariant for bisimulation \cite{bent:moda83}. Here, 
due to the backward modal operators and the global modal operator,
we need global two-way bisimulations (see for example~\cite{Otto2004}).
Given two Kripke structures
$M$ and $N$ a \emph{global two-way bisimulation} between $M$ and $N$ is a binary relation $Z\subseteq
M\times N$ such that the following hold for every pair $(a,b)\in Z$ and every
relation symbol $R$:
\begin{itemize}
\item $a\in P^M$ if and only if $b\in P^N$, for all unary relation
  symbols $P$.
\item for every $a'$ with $(a,a')\in R^M$ there is a 
  $b'$ such that $(b,b')\in R^N$ and $(a',b')\in Z$,
\item for every $b'$ with $(b,b')\in R^N$ there is an
  $a'$ such that $(a,a')\in R^M$ and $(a',b')\in Z$,
\item for every $a'$ with $(a',a)\in R^M$ there is a
  $b'$ such that $(b',b)\in R^N$ and $(a',b')\in Z$,
\item for every $b'$ with $(b',b)\in R^N$ there is an
  $a'$ such that $(a',a)\in R^M$ and $(a',b')\in Z$,
\item for every node $a'$ of $M$ there is a node
  $b'$ of $N$ such that $(a',b')\in Z$,
\item for every node $b'$ of $N$ there is a node
  $a'$ of $M$ such that $(a',b')\in Z$,
\end{itemize}
We write $M\bisim N$ if there is a global two-way bisimulation between
$M$ and $N$, and we write $(M,a)\bisim (N,b)$ if the pair $(a,b)$ belongs
to a global two-way bisimulation between $M$ and $N$. Recall that a
homomorphism $h:M\to N$ is a map from the domain of $M$ to the domain
of $N$ such that for all relation symbols $R$ and tuples $(a_1,
\ldots, a_n)\in R^M$, we have that $(h(a_1),
\ldots, h(a_n))\in R^N$. We say that $M$
is a \emph{$\bisim$-cover} of $N$ if there is a homomorphism $h:M\to N$ such 
that  $(M,a)\bisim (N,h(a))$ for every element $a$ of $M$. In addition,
for  $h(a')=a$, we say that $(M,a')$ is a $\bisim$-cover of $(N,a)$. 

One can equivalently view bisimulations as strategies for a player in a
two-player game. In this game, the two players, called Abelard en
Elo{\"i}se, maintain a pair $(a,b)$ of elements, one in each
structure. Intuitively, one can think of a pebble lying on each of
these two selected nodes.  At any time during the game, the pebbled
nodes must satisfy the same unary predicates in the two structures. A
move of Abelard consists of choosing one of the two pebbles (i.e., one
of the two structures), and either sliding the pebble forward or
backward along an edge belonging to some binary relation $R$, or
moving it to an arbitrary position.  Then Elo\"ise must respond in the
other structure with a move mimicking Abelard's move, that is, either
sliding her pebble along an edge belonging to the same relation $R$
(and in the same
direction), or moving it to an arbitrary position, depending on Abelard's
move. If Elo\"ise cannot respond with a valid move, Abelard wins.  It
is easy to see that $(M,a)\bisim (N,b)$ if and only if, starting in
$(a,b)$, Elo\"ise has a strategy that allows her to play forever
without letting Abelard win.  We also write $(M,a)\bisim^l (N,b)$ if,
starting in $(a,b)$, Elo\"ise has a strategy that avoids losing in the
first $l$ rounds.

It is well known that $\ml_\mu$-formulas are invariant for global
two-way bisimulations: if $(M,a)\bisim (N,b)$ and if $\phi$ is a formula
of $\ml_\mu$, then $(M,a)\models\phi$ if and only if $(N,b)\models
\phi$. This basic fact of $\ml_\mu$ has an important consequence: if a
\mucalc formula has a model, then it has a, possibly infinite, acyclic
model, obtained by ``unraveling'' the original model along its paths
while preserving bisimulation equivalence. 
If we restrict attention to finite Kripke structures, then, in general,
acyclicity may not be achievable. For example, the one-element structure
consisting in a self-loop is not bisimilar to any finite acyclic structure.
However, over finite structures, a weaker form of acyclicity can be
achieved. For a natural number $l$, a Kripke structure is called
\emph{$l$-acyclic} if its underlying graph contains no cycle of length less
than $l$. We will make use of the following important result:

\begin{theorem}\label{thm-otto-acyclic-cover}\cite{Otto2004}
  For all $l\in\mathbb{N}$, every finite Kripke structure has a finite
  $l$-acyclic $\bisim$-cover.
\end{theorem}

This was used to show the following property, which is relevant for us
as well. We write $(M,a)\equiv_{\fo^q} (N,b)$ if $(M,a)$ and $(N,b)$
satisfy the same first-order formulas of quantifier depth $q$. 

\begin{theorem}\label{thm-otto-cover}\cite[Proposition 33]{Otto2004}
  For each $q\in\mathbb{N}$ there is a $l\in\mathbb{N}$ such that whenever
 $(M,a) \bisim^l (N,b)$, then $(M,a)$ and $(N,b)$
   have $\bisim$-covers $(M',a')$ and $(N',b')$, respectively, such that 
 $(M',a') \equiv_{\fo^q} (N',b')$. Moreover, if $M$ and $N$ are finite then $M'$
 and $N'$ can be chosen to be finite as well.
\end{theorem}

\subsection{Overview of relevant oracle complexity  classes} \label{section-complexity}

We briefly review a number of complexity classes involving restricted
access to an oracle, which turn out to be relevant for our present
investigation, and may not be very well known. The reader interested
to learn more about these classes would benefit from reading the
literature cited below, as well as~\cite{Schnoebelen03:oracle}
and~\cite{Gottlob95:np} that inspired us a lot.

The first class we use is denoted \pnp, also known as $\Delta^p_2$. It consists of all problems that
are computable by a Turing machine running in time polynomial in the size of
its input, where the Turing machine, at any point during its computation, can
ask yes/no queries to an \np oracle, and take the answers of the oracle into
account in subsequent steps of the computation (including subsequent queries to
the \np oracle).  Analogously, one can define the classes $\npnp$ and
$\conpnp$, which are also known as $\Sigma^p_2$ and $\Pi^p_2$, respectively.
 An example of a \pnp-complete problem is LEX(SAT), which takes as input a
 Boolean formula $\phi(x_1, \ldots x_n)$ and asks what is the value of
 $x_n$ in the lexicographically maximal solution (where $x_n$ is treated
 as the least significant bit in the ordering)~\cite{Wagner:BH}.

 A subclass of \pnp is \pnplog. It is defined in the same way as \pnp, except
 that the number of yes/no queries that can be asked to the \np oracle is
 bounded by $O(\log(n))$, where $n$ is the size of the input. There is an
 equivalent characterization of this class, denoted \pnppar
 in~\cite{Buss91:truth-table}. It consists of all problems computable using a
 Turing machine running in time polynomial in the size of its input, where the
 Turing machine can call the oracle only once (or a constant number of
 times, as this turns out not to make a difference), but in
 doing so it may ask the oracle several (polynomially many) yes/no questions in
 parallel~\cite{Buss91:truth-table}. In other words, the answer to a query cannot be used by the Turing
 machine in choosing which subsequent queries to ask to the oracle.  A third equivalent
 characterization of \pnplog is as the class of problems that are \ptime
 truth-table reducible to \np~\cite{Buss91:truth-table}, that is, problems for which
 there is a \ptime algorithm that, given an instance of the problem, produces a
 set $y_1,\cdots,y_n$ of inputs to some \np oracle, together with a Boolean
 formula $\phi(x_1,\cdots,x_n)$, such that the input is a yes-instance iff $\phi$
 evaluates to true after replacing each $x_i$ by $1$ if $y_i$ is accepted by
 the \np oracle and $0$ otherwise.
 An example of a \pnplog-complete problem is the problem whether 
 two graphs have the same chromatic number~\cite{Wagner:BH}.

Finally, in between \pnplog  and \pnp lies a hierarchy of classes \pnplogi with
$i> 1$. They are defined in the same
way as \pnplog except that the number of queries to the oracle is 
bounded by $O(\log^i(n))$. Each class \pnplogi can be equivalently characterized as
the class of problems that can be solved in polynomial time allowing
$O(\log^{i-1}(n))$ many rounds of parallel queries to an \np oracle \cite{CastroSeara1996}.

There are few known natural complete problems for the classes
\pnplogi. We introduce here a complete problem \lexisat that we will make
use of later on in our lower bound proofs.  Recall that \lexsat is the
problem to decide, given a Boolean formula $\phi(x_1,
\ldots, x_n)$, whether the value of $x_n$ is 1 in the
lexicographically maximal solution. Here, $x_1$ is treated as the most
significant bit and $x_n$ as the least significant bit in the
ordering. Similarly, for $i\geq 1$, we define \lexisat to be the
problem of testing, given a Boolean formula $\phi(x_1,
\ldots, x_n)$ and a number $k\leq \log^i(n)$, whether the value of $x_k$ is
1 in the lexicographically maximal solution.
\begin{theorem}
\lexisat is \pnplogi-complete. 
\end{theorem}
\begin{proof}
  The upper bound proof is immediate. In order to test whether the value of $x_1$  is 1
  in the lexicographically maximal solution it is enough to ask the oracle
  whether the Boolean formula $\phi(1,x_2,\ldots, x_n)$ has a solution or
  not. Depending on the result we continue with $\phi(1,x_2,\ldots, x_n)$  or
  $\phi(0,x_2,\ldots, x_n)$ and with $j$ calls to the oracle we learn this way the
  value of $x_1,\ldots,x_j$ in the lexicographically maximal solution.

  The proof for the lower bound is a straightforward adaptation of the proof in
  \cite{Wagner:BH} that \lexsat  is \pnp-complete.
  Let $A$ be any problem in \pnplogi, let $M$ be a deterministic
  polynomial-time Turing machine accepting $A$ using an oracle
  for a problem $B\in\np$, and let $M'$ be a non-deterministic
  polynomial-time Turing machine accepting $B$. Let 
  $f(n)= O(\log^i(n))$ be a function
  bounding the number of oracle queries asked by $M$ on an input of
  size $n$.

  Recall the textbook proof of NP-hardness of propositional
  satisfiability (cf., for example, \cite{WagnerWechsung}), which is
  based on efficiently constructing a Boolean formula describing runs of a given
  non-deterministic polynomial-time Turing machine. By the same construction, we
  can efficiently compute a Boolean formula
  $$H^n(\tuple{x},\tuple{u},\tuple{v}_1, \ldots, \tuple{v}_{f(n)}, z_1, \ldots, z_{f(n)},
  z)$$ (in fact, a conjunction of clauses with 3 literals per clause),
  parametrized by a natural number $n$,
  whose satisfying assignments describe the runs of $M$ on input words
  of length $n$,
  where
\begin{enumerate}
\item $\tuple{x}$ describes an input word of length at most $n$
\item $\tuple{u}$ describes the sequence of configurations of $M$
   during the run (including tape content, head position and state in
   each configuration),
\item $\tuple{v}_j$ describes the $j$-th query asked to the oracle,
\item $z_j$ describes the answer of the $j$-th query asked to the
  oracle (where $z=0$ means ``no'' and $z=1$ means ``yes''),
\item $z$ describes the result of the entire computation of $M$ (where $z=0$ 
  means ``reject'' and $z=1$ means ``accept'').
\end{enumerate}
The formula $H^n(\tuple{x},\tuple{u},\tuple{v}_1, \ldots, \tuple{v}_{f(n)}, z_1,
\ldots, z_{f(n)}, z)$ does \emph{not} enforce that each $z_j$ is the
\emph{correct} answer to the oracle query $\tuple{v}_j$. Thus, 
the formula may have several satisfying assignments with the same values
of $\tuple{x}$ (one for each
possible sequence of answers that the oracle may give).

In the same way, we can efficiently compute a Boolean formula 
$$G^n(\tuple{v},\tuple{y},z')$$ 
(in fact, a conjunction of clauses with 3 literals per clause),
  whose satisfying assignments describe all (not necessarily accepting) runs of $M'$ on input words
  of length $n$,
  where
\begin{enumerate}
\item $\tuple{v}$ describes an input word of length at most $n$
\item $\tuple{y}$ describes the sequence of configurations of $M'$
   during the run (including tape content, head position and state in
   each configuration),
\item $z'$ describes the result of the entire computation of $M$ (where $z=0$ 
  means ``reject'' and $z=1$ means ``accept'').
\end{enumerate}
Since $M'$ is non-deterministic, the formula
$G^n(\tuple{v},\tuple{y},z')$ may have many satisfying assignments with
the same values of $\tuple{v}$. 

Finally, for each input word $w$ given as a bitstring of length $n$, we define $\phi_w$ to
be the Boolean formula 
$$\phi_w = H^n(w,\tuple{u},\tuple{v}_1, \ldots, \tuple{v}_{f(n)}, z_1,
\ldots, z_{f(n)}, z)\land\bigwedge_{j=1\ldots f(n)} G^n(\tuple{v}_j,\tuple{y}_j,z_j)$$
Observe that $\phi_w$ only asserts that the $z_j$ is the result of 
\emph{some} run of the non-deterministic Turing machine $M'$ on input $\tuple{v}_j$. It
does not require that $z_j=1$ when $M'$ has an accepting run
 on  input $\tuple{v}_j$. Consequently, not every satisfying
 assignment
of $\phi_w$ describes the correct computation of $M$ on input
$\tuple{x}$.  However, it is easy to see that the lexicographically 
maximal satisfying assignment \emph{does}  describe the correct
computation, due to the fact that it makes $z_j=1$ whenever possible
(given the already obtained values for $z_\ell$ for $\ell<j$). Thus, 
we have
$z=1$ in the lexicographically maximal solution of
  $\phi_w$ if and only if $w\in A$.

  We order the variables in the formula $\phi_w$ so that $z_1, \ldots,
  z_{f(n)}, z$ come first (and in this order).  By construction, $z$
  is then the $f(n)+1$st variable of $\phi_w$.  Let $m$ be the total
  number of variables occurring in $\phi_n$ (which is bounded by
  some polynomial in $n$).  If $f(n)+1\leq \log^i m$, then
  $(\phi_w,m)$ is a valid input for the \lexisat problem, and we
  are done. Otherwise, we extend $\phi_w$  with 
  additional dummy variables, which serve no role other than making
  sure that $f(n)+1\leq \log^i m$. It is easy to see that this can
  always
  be done.
\end{proof}

\section{Model theory}\label{section-express}
\newcommand\type{\text{type}}
\newcommand{\unhom}{\to_{\textup{UN}}}
\newcommand{\unkhom}{\to_{\textup{UN}^k}}

In this section we give many key definitions, we show results about the
expressive power of \unfo and \unfp, and we show that \unfo has Craig
Interpolation and the Projective Beth Property.

\subsection{UN-bisimulations, the finite model property, and the
  tree-like model property} \label{sec:unbisimulations}

%
We define a game that captures model indistinguishability, and we use it to
characterize the expressive power of \unfo and \unfp.  The game is as follows:
the two players maintain a single pair $(a,b)$ of elements from the two
structures. A move of Abelard consists of choosing a set $X$ of points in one
of the two structures.  Then Eloise responds with a homomorphism $h$ from the
set $X$ into a set of points in the other structure, where the homomorphism maps
$a$ to $b$ (respectively $b$ to $a$) if $a$ (respectively $b$) belongs to the
set $X$.  Finally, Abelard picks a pair $(u,h(u))$ (respectively $(h(u),u)$) and the players
continue with that pair. The game is parametrized by the size of the sets
chosen by Abelard in each round.

Equivalently, we can present the game in terms of a back-and-forth
system:

\begin{definition}\label{def-bisimulation}
  Let $M,N$ be two structures. A UN-bisimulation (resp. a UN-bisimulation of width $k \geq 1$)
  is a binary relation $Z\subseteq M\times N$ such that the following hold for
  every pair $(a,b)\in Z$:

\begin{itemize}
\item {[{\bf Forward property}]} For every finite set $X\subseteq \dom(M)$ (resp. with $|X|\leq k$)
  there is a partial homomorphism $h:M\to N$ whose domain is $X$,
  such that $h(a)=b$ if $a\in X$, and such that every pair $(a',b')\in h$ belongs to $Z$.
\item {[{\bf Backward property}]} Likewise in the other direction, where
  $X\subseteq \dom(N)$.
\end{itemize}
We write $M\unbisim N$ if there is a non-empty UN-bisimulation between
$M$ and $N$, and we write $M\unkbisim N$ if there is a
non-empty UN-bisimulation of width $k$ between $M$ and $N$.
\end{definition}

It is not difficult to see that the existence of a UN-bisimulation implies
indistinguishability by \unfp sentences, and that the (weaker) existence of a
UN-bisimulation of width~$k$ implies indistinguishability in
$\unfp^k$. 

\begin{prop}\label{prop-unbisim}
  For any $k\geq 1$, if $M\unkbisim N$ then $M$ and $N$ satisfy the
  same sentences of $\unfp^k$. In particular, if $M\unbisim
  N$ then $M$ and $N$ satisfy the same sentences of $\unfp$.
\end{prop}

\begin{proof}
  The second claim follows immediately from the first one, because 
  $M\unbisim N$ implies $M\unkbisim N$ for all $k\geq 1$. 

  The proof of the first claim is by induction on the nesting of fixpoints and existential
  quantification in the formula. We assume without loss of generality that all
  formulas are in \unnormalform.  It is convenient to state the induction
  hypothesis for $\unfo^k$-formulas $\phi(x)$ in one free first-order variable and several free
  monadic second-order variables.  The induction hypothesis then
  becomes: for all formulas
  $\phi(x,\tuple{Y})$ of width $k$, for all UN-bisimulations $Z$ of width $k$ between
  $(M,\tuple{P})$ and $(N,\tuple{Q})$, and for  all pairs $(a,b) \in Z$, we have
  $(M,\tuple{P},a)\models\phi$ iff $(N,\tuple{Q},b)\models\phi$.  We show only
  the important cases of the inductive step.  Let $M,N$ be two structures, $Z$
  be a UN-bisimulation of width $k$ between $M$ and $N$, $\tuple{P}$ and
  $\tuple{Q}$ be valuations of $\tuple{Y}$ respectively on $M$ and $N$, and
  $(a,b)\in Z$.

\begin{itemize}
\item $\phi(x,\tuple{Y})$ starts with an existential quantifier.  Then, by
  definition of \unnormalform, $\phi$ starts with a block of existential
  quantifications, followed by a Boolean combination of atomic formulas or
  formulas with at most one free first-order variable, i.e. again of the form
  $\psi(y,\tuple{Y})$.  Let $x_1, \ldots, x_n$ be the initial existentially
  quantified variables of $\phi$. In particular $n \leq k$.

  First, suppose $(M,\tuple{P},a)\models \phi$. Let $X=\set{a_1, \ldots, a_n}$
  be the quantified elements of $M$ witnessing the truth of $\phi$.  By the
  definition of UN-bisimulation, there is a homomorphism $h:M\restrict X \to N$
  such that $h(a)=b$ (if $a$ is in the domain of $h$) and such that
  $\{(a_i,h(a_i))\mid i\leq n\}\subseteq Z$. By induction hypothesis, a subformula
  $\psi(y,\tuple{Y})$ of $\phi$ is true on $(M,\tuple{P},a_i)$ iff it is true on
  $(N,\tuple{Q},h(a_i))$.  Hence the assignment that sends $x_1, \ldots, x_n$ to
  $h(a_1),\ldots, h(a_n)$ makes $\phi$ true on $(N,\tuple{Q},b)$. The opposite direction,
  from $(N,\tuple{Q},b)\models\phi$ to $(M,\tuple{P},a)\models\phi$, is symmetric.

 \item $\phi(x,\tuple{Y})$ is any Boolean combination of formulas of the form
   $\psi(y,\tuple{Y})$, the result is immediate from the induction hypothesis.

 \item $\phi(x,\tuple{Y})$ is of the form
   $[\LFP_{X,y}~\psi(X,\tuple{Y},y)](x)$. We proceed by induction on
   the the fixpoint iterations. Let $\mathcal{O}_{\phi,(M,\tuple{P})}$
   and $\mathcal{O}_{\phi,(N,\tuple{Q})}$ be be the monotone
   set-operations induced by $\phi$ on subsets of the domain of
   $(M,\tuple{P})$ and $(N,\tuple{Q})$, respectively, and let $\kappa
   = \max\{|M|,|N|\}$. Recall that the least fixpoint of 
$\mathcal{O}_{\phi,(M,\tuple{P})}$ is equal to 
${\mathcal{O}_{\phi,(M,\tuple{P})}}^\kappa(\emptyset)$, and similarly
for the least fixpoint of 
$\mathcal{O}_{\phi,(N,\tuple{Q})}$. A
   straightforward transfinite induction shows that, for all
   ordinals $\lambda$, and for all $(a,b)\in Z$, 
   $a\in {\mathcal{O}_{\phi,(M,\tuple{P})}}^\lambda(\emptyset)$ if and only if $b\in
   {\mathcal{O}_{\phi,(N,\tuple{Q})}}^\lambda(\emptyset)$.  We
   conclude
   that $(M,\tuple{P},a)\models  [\LFP_{X,y}~\psi(X,\tuple{Y},y)](x)$
   if and only if $(N,\tuple{Q},b)\models  [\LFP_{X,y}~\psi(X,\tuple{Y},y)](x)$.\qedhere
\end{itemize}
\end{proof}

\noindent Note that it is crucial, here, that we have defined width in terms of the
\unnormalform. For example, if $R$ is a binary relation, then 
the existence
of a cyclic directed $R$-path of length $k$ (i.e., a sequence of not
necessarily distinct nodes $a_1, \ldots, a_k$ with $R(a_i,a_{i+1})$
and $R(a_k,a_1)$) can be expressed in \unfo using
only 3 variables, by means of a careful reuse of variables, but the
formula in question would not be in \unnormalform. Indeed, 
the existence of a cyclic directed $R$-path of length $k$, for $k>3$, is not
preserved by UN-bisimulations of width~$k-1$. 

A similar invariance property holds for formulas with free
variables. For simplicity, we only state a version of the result without
reference to the width of formulas. 
\begin{definition}
 Let $M$ and $N$ be structures with the same signature. 
 A \emph{UN-homomorphism}
$h:M\to N$ is a homomorphism with the property that $(M,a)\unbisim
(N,h(a))$ for all $a\in \dom(M)$. We write $(M,\tuple{a})\unhom
(N,\tuple{b})$ if there is a UN-homomorphism $h:M\to N$ such that 
$h(\tuple{a})=\tuple{b}$. 
\end{definition}

\begin{prop}\label{prop-unhom}
If $(M,\tuple{a})\unhom (N,\tuple{b})$ and $M\models\phi(\tuple{a})$
then $N\models\phi(\tuple{b})$, for all \unfp-formulas
$\phi(\tuple{x})$. 
\end{prop}

\begin{proof} Follows from Proposition~\ref{prop-unbisim}, together
 with the fact that positive existential  formulas are preserved
 by homomorphisms (note that every \unfp-formula $\phi(\tuple{x})$
 can be viewed as
 a positive existential formula built from atomic formulas and from \unfp-formulas in
 one free variable).  
\end{proof}

%
From the invariance for UN-bisimulation it follows by a standard infinite
unraveling argument that \unfp has the tree-like model property.  A more
involved partial unraveling, using back-edges in order to keep the structure
finite, can also be used to show that \unfo has the finite model property.  We
only state the results without giving the details of these constructions, as it turns out that both
results will follow from the material presented in Section~\ref{section-sat}.

\begin{theorem}\label{thm-finite-model-prop}
  Every satisfiable \unfo formula has a finite model.
\end{theorem}

\begin{theorem}\label{thm-tree-like-prop}
 Every satisfiable \unfp formula of \width $k$  has a model 
 of tree-width $k-1$.
\end{theorem}

Note, that \unfp does \emph{not} have the finite model
property. This follows from the fact that \unfp contains the two-way
$\mu$-calculus which is known to lack the finite model
property~\cite{Streett82}. Indeed, if $\max(x)$ is shorthand for $\neg\exists y ~E(x,y)$, then 
the formula 
\[ \exists x \max(x) \lor \exists x \neg[\LFP_{X,y}~\neg\exists
z(E(z,y) \land \neg X(z))](x)\]
expresses the property that either there exists a maximal element or there is
an infinite backward path. This formula is therefore obviously false in
the infinite structure $(\mathbb{N},suc)$. However it holds on any finite
structure as if a finite structure has no maximal elements,
it must contain a cycle, and hence an infinite backward path. The negation of
this sentence is satisfiable, by $(\mathbb{N},suc)$, but has no finite model.

\subsection{Characterizations}
%
We have seen in Proposition~\ref{prop-unbisim} that \unfo sentences are
first-order formulas that are preserved under \unbisim-equivalence.  It turns out
that the converse is also true. Indeed, in the same way that
bisimulation-invariance characterizes modal
logic~\cite{bent:moda83,rose:moda97} and guarded bisimulation-invariance
characterizes the guarded fragment of \fo~\cite{andr:moda98,Otto12}, we will see that
$\unbisim$-invariance characterizes $\unfo$.  We show two variants of this
result depending on whether we consider finite or infinite structures.
It turns out that the proof for the finite case also works for the infinite
case. However we give an independent proof for the infinite case as it is simpler
and introduces techniques that will be useful later when considering Craig
Interpolation.

We say that a \fo sentence is \emph{\unbisim-invariant} if for all structures
$M$ and $N$ such that $M\unbisim N$, we have $M\models\phi$ iff
$N\models\phi$. The notion of $\unkbisim$-invariance is defined similarly.

\begin{theorem}\label{thm-charac-fo}
A sentence of \fo is equivalent to a formula of \unfo iff it is
$\unbisim$-invariant.

For all $k\geq 1$, a sentence of \fo is equivalent to a formula of $\unfo^k$ iff it is
$\unkbisim$-invariant.
\end{theorem}

Before proving Theorem~\ref{thm-charac-fo} we state and prove the
following useful lemma. We write $(M,a)\equiv_{\unfo} (N,b)$ if
$(M,a)$ and $(N,b)$ satisfy the same \unfo-formulas. We define
$(M,a)\equiv_{\unfo^{k}} (N,b)$ similarly.  This lemma makes use of the
classical notion of $\omega$-saturation. The actual definition is not
needed here and the interested reader is referred
to~\cite{hodg:mode93}. For our purpose it is enough to know the
following two key properties.
\begin{enumerate}
\item  If a (possibly infinite) set of
first-order formulas is satisfiable, then it is satisfied by a model
that is $\omega$-saturated.
\item Let $M$ be $\omega$-saturated, let $a_1, \ldots, a_m\in
  \dom(M)$, and let $T(x_1, \ldots, x_n)$ is an infinite set of
  first-order formulas with free variables $x_1, \ldots, x_n$ and
  using $a_1, \ldots, a_m$ as parameters. If every finite subset
  $T'$ of $T$ is realized in $M$ (meaning that $(M,\tuple{b},\tuple{a})
  \models T(\tuple{x})$ for some $\tuple{b}=b_1, \ldots, b_n\in \dom(M)$), then the entire set $T$ is realized in $M$.
\end{enumerate}

\begin{lemma}\label{lem:saturated}
For all $\omega$-saturated structures $M$ and $N$ with elements
  $a$ and $b$, respectively, the following hold.
\begin{enumerate}
\item The relation $\{(a,b)\mid (M,a)\equiv_{\unfo}(N,b)\}$ is a
  UN-bisimulation.
\item  The relation $\{(a,b)\mid (M,a)\equiv_{\unfo^k}(N,b)\}$ is a
  UN-bisimulation of width $k$ ($k\geq 1$).
\end{enumerate}
\end{lemma}

\begin{proof}
  We prove the second claim. The proof of the first claim is similar.
  Let $Z = \{(a,b)\mid (M,a)\equiv_{\unfo^k}(N,b)\}$.
  We show that $Z$ satisfies the forward property, the proof of the
  backward property is analogous. 

  Suppose $(c,d)\in Z$ and let $X\subseteq \dom(M)$ with $|X|\leq k$.  We can
  distinguish two cases: either $c\in X$ or $c\not\in X$. We will consider the
  first case (the second case is simpler). Thus, let $X=\{c,c_1, \ldots, c_n\}$
  ($n<k$). Let $T[x_1, \ldots, x_n]$ be the set of all formulas $\phi(x,x_1,
  \ldots, x_n)$ that are positive Boolean combinations of atomic formulas or
  unary formulas of $\unfo^k$ and that are true in $(M,c, c_1, \ldots, c_n)$. We
  view $T$ as an $n$-type with one parameter. 

  Notice that by construction of $T$, for each finite subset $T'$ of $T$ the
  formula $\exists x_1\ldots x_n(\bigwedge T')$ is in $\unfo^k$ and is
  satisfied by $(M,c)$. By hypothesis this formula is therefore also satisfied
  by $(N,d)$. Since $N$ is $\omega$-saturated (and treating $T$ as an $n$-type
  with parameter $d$), it follows that the entire set $T[x,x_1, \ldots, x_n]$
  is realized in $N$ under an assignment $g$ that sends $x$ to $d$.  This
  implies that the function $h$ sending $c$ to $d$ and $c_i$ to $g(x_i)$ is
  a homomorphism such that, for all $i$, $c_i$ and $h(c_i)$ satisfy the same
  formulas of $\unfo^k$ and therefore $(c_i,h(c_i)) \in Z$ by definition
  of~$Z$.
\end{proof}

\begin{proof}[Proof of Theorem~\ref{thm-charac-fo}]
  One direction follows from Proposition~\ref{prop-unbisim}. For the other
  direction, we only give the proof for the case of $\unfo^k$, the argument for
  full $\unfo$ being identical.

  Let $\phi$ be any $\unkbisim$-invariant \fo sentence. We want to
  show that $\phi$ is equivalent to a $\unfo^k$-sentence.

  We first show that whenever two structures agree on all sentences of
  $\unfo^{k}$, they agree on $\phi$. Suppose $M$ and $N$ satisfy the same
  sentences of $\unfo^{k}$. Without loss of generality we can assume that $M$
  and $N$ are $\omega$-saturated. Define $Z\subseteq M\times N$ as the set of
  all pairs $(a,b)$ such that $(M,a)$ and $(N,b)$ satisfy the same
  $\unfo^k$-formulas. By Lemma~\ref{lem:saturated}, $Z$ is a UN-bisimulation of
  width $k$. We claim that $Z$ is non-empty. Let $a$ be any element of $M$, and
  let $\Sigma(x)$ be the set of all $\unfo^k$-formulas with one free variable,
  true for $a$ on $M$. Notice that for every finite subset $\Sigma'$ of $\Sigma$,
  the formula $\exists x ~\bigwedge\Sigma'$ is a sentence of $\unfo^{k}$ that is
  satisfied by $M$. Hence, by hypothesis, it is also satisfied in
  $N$. Therefore, by $\omega$-saturation, the entire set $\Sigma(x)$ is
  realized by an element $b$ in $N$, and hence $(a,b)\in Z$, which implies that
  $Z$ is non-empty. This implies that $M \unkbisim N$. Assume now that
  $M\models \phi$. As $\phi$ is $\unkbisim$-invariant and $M \unkbisim N$, this
  implies that $N \models \phi$. By symmetry we get $M\models \phi$ iff
  $N\models \phi$ as desired.

  The rest of the proof is a well known argument using Compactness: If $\phi$
  is not satisfiable then $\phi$ is equivalent to the $\unfo^1$ sentence
  ``false''. Otherwise let $M \models \phi$ and let $\Theta$ be the set of all
  $\unfo^k$ sentences $\theta$ such that $M \models \theta$. We show that
  $\Theta \models \phi$ (i.e. any model of $\Theta$ is a model of $\phi$). If
  this were not the case then we have a structure $N$ such that $N\models
  \Theta\land\lnot\phi$. But because $\Theta$ contains each $\unfo^k$ sentence
  or its negation we have $M\unkbisim N$ and $M,N$ disagree on $\phi$. This
  contradict the claim of the previous paragraph.

  By compactness, there is a finite subset $\Theta'$ of $\Theta$ such that
  $\Theta' \models \phi$. By construction, this implies that $\phi$ is equivalent to the
  conjunction of all the sentences in $\Theta'$.
\end{proof}

Before we turn to the finite variant of Theorem~\ref{thm-charac-fo},
we remark
that a similar characterization can be obtained for formulas with free
variables, using UN-homomorphisms instead of UN-bisimulations.

\begin{theorem}\label{thm-charac-fo-general}
A formula of \fo with free variables is equivalent to a formula of \unfo iff it is
preserved under UN-homomorphisms.
\end{theorem}

\begin{proof}
  One direction is provided by Proposition~\ref{prop-unhom}. For the other
  direction, let $\phi(\tuple{x})$ be a \fo formula preserved under
  UN-homomorphisms.  As for the proof of Theorem~\ref{thm-charac-fo}, using a
  standard Compactness argument (cf.~\cite[Lemma 3.2.1]{ChangKeisler}), it is enough to show that, for structures
  $M,N$ with tuples $\tuple{a}, \tuple{b}$, if every \unfo-formula true in
  $(M,\tuple{a})$ is true in $(N,\tuple{b})$, then also $M\models\phi(\tuple{a})$
  implies $N\models\phi(\tuple{b})$. Without loss of generality we may assume
  that $M$ and $N$ are $\omega$-saturated. As $\phi$ is preserved under
  UN-homomorphisms, the result is now a direct consequence of the following
  claim.

\begin{quote}
  If $M,N$ are $\omega$-saturated structures, and $\tuple{a},
  \tuple{b}$ tuples of elements such that every \unfo-formula true in
  $(M,\tuple{a})$ is true in $(N,\tuple{b})$, then there is a
 UN-homomorphism from an elementary substructure of $(M,\tuple{a})$
  to an elementary substructure of $(N,\tuple{b})$ that maps $\tuple{a}$ to
  $\tuple{b}$.
\end{quote}
%
  In what follows, we prove the above claim.
  First of all, note that every equality statement satisfied in
  $\tuple{a}$ is satisfied in $\tuple{b}$, which makes it meaningful
  to speak about functions mapping $\tuple{a}$ to $\tuple{b}$.

  We need the notion of a \emph{potential homomorphism} from a
  structure $M$ to a structure $N$. It is a non-empty collection
  $\mathcal{F}$ of finite partial homomorphisms $f:M\to N$, which
  satisfies the following extension property: for all
  $f\in\mathcal{F}$ and for all $a\in \dom(M)$, there is an
  $f'\in\mathcal{F}$ which extends $f$ and whose domain includes $a$.
  By a \emph{potential UN-homomorphism} we will mean a potential
  homomorphism whose finite partial homomorphisms preserve
  the UN-bisimilarity type of each node. 
  It is not hard to see that if $M,N$ are countable structures and
  $\mathcal{F}$ is a potential homomorphism from $M$ to $N$, then
  there is a homomorphism $h:M\to N$, which can be defined as the limit
  of a sequence of finite partial homomorphisms belonging to
  $\mathcal{F}$.  In particular, if $\mathcal{F}$ is a potential
  UN-homomorphism, then $h$ is a UN-homomorphism.

  Now, let $M$ and $N$ be structures as described by the statement of
  the Lemma.  A straightforward variation of the proof of
  Lemma~\ref{lem:saturated} shows that there is a potential
  UN-homomorphism $\mathcal{F}$ from $M$ to $N$ mapping $\tuple{a}$ to
  $\tuple{b}$.

  Next, we take the model pair $(M,N,\tuple{a},\tuple{b})$ expanded with
    the maximal UN-bisimulation relation $Z$ between $M$ and $N$
    (i.e., the binary relation containing all pairs $(a',b')$ such
    that $(M,a')\unbisim (N,b')$), plus
    infinitely many additional relations that represent the potential
    UN-homomorphism $\mathcal{F}$ (for each $k\leq 1$, we use a new
    $2k$-ary relation $R_k$ to represent all finite partial
    homomorphisms defined on $k$ elements). We then apply to downward
    Löwenheim-Skolem theorem to obtain a similar situation
    $(M',N',\tuple{a},\tuple{b}, \ldots)$, but where $M'$ and $N'$ are
      countable elementary substructures of $M$ and $N$. 
    Since we added the $Z$ and $R_k$ relations 
    before applying the L\"owenheim-Skolem theorem, we still have a potential
  UN-homomorphism from $M'$ to $N'$ mapping $\tuple{a}$ to
  $\tuple{b}$.
  It then follows by the earlier remark that there is
      a UN-homomorphism from $M'$ to $N'$ mapping $\tuple{a}$ to
      $\tuple{b}$. 
\end{proof}

Finally, we consider the case of finite structures.  We say that a \fo sentence
is \emph{\unbisim-invariant on finite structures} if for all finite
structures $M$ and $N$ such that $M\unbisim N$, we have $M\models\phi$ iff
$N\models\phi$. The notion of \emph{$\unkbisim$-invariance on finite structure} is
defined in the same way.
We prove that $\unfo^k$ is also the $\unkbisim$-invariant fragment of \fo on
finite structures. The proof of this result constitutes, at the same
time, also an alternative
proof of the second part of Theorem~\ref{thm-charac-fo}. However it relies in a
crucial way on the parameter $k$ and therefore does not yield a
characterization of \unfo in terms of
$\unbisim$-invariance. In particular, it remains open whether $\unfo$ is the
$\unbisim$-invariant fragment of \fo on finite structures. For
simplicity, we state the result for formula with at most one free variable.

\begin{theorem}\label{thm-charac-fo-finite}
  Fix a finite schema $\sigma$, and let $m$ be the maximal arity of the
  relations in $\sigma$.  For all $k\geq m$, a formula of \fo with at most one
  free variable is equivalent over finite $\sigma$-structures to a sentence of
  $\unfo^k$ iff it is $\unkbisim$-invariant on finite $\sigma$-structures.
\end{theorem}

\begin{proof}
  Recall that we write $(M,a)\equiv_{\fo_q} (N,b)$ if $M$ and $N$ satisfy the
  same \fo sentences of quantifier depth $q$.  Recall the definition of
  \unnormalform, and observe that, in \unfo formulas that are in \unnormalform
  and that have at most one free variable, every existential quantifier must be
  either directly below another existential quantifier, or, otherwise, the
  subformula starting with that quantifier has at most one free variable. In
  the latter case, we call the existential quantifier in question a
  \emph{leading existential quantifier}. We say that a formula of $\unfo^k$ in
  \unnormalform with at most one free variable has \emph{block depth $q$} if
  the nesting depth of its leading quantifiers is less than $q$. We denote by
  $\unfo^k_q$ the fragment of $\unfo^k$ consisting of formulas in \unnormalform
  with at most one free variable that have block depth $q$, and we will write
  $(M,a)\equiv_{\unfo^k_q} (N,b)$ if $M$ and $N$ satisfy the same formulas of
  $\unfo^k_q$.  One can show by a straightforward induction on $k$ that, for
  every fixed finite schema, $\equiv_{\unfo^k_{q}}$ is an equivalence relation
  of finite index (that is, there are only finitely many equivalence classes)
  and, consequently, each equivalence class can be described using a single
  $\unfo^k_{q}$-formula.

The basic proof strategy is as follows: we will show that for every $q\geq 0$
there exists a $l\geq 0$ such that whenever $(M,a)\equiv_{\unfo^k_{l}} (N,b)$, then
there exist structures $(M^*,a^*)$ and $(N^*,b^*)$ such that $(M,a)\unkbisim (M^*,a^*)$,
$(M^*,a^*)\equiv_{\fo^q} (N^*,b^*)$, and $(N^*,b^*)\unkbisim (N,b)$. In other words, using the
terminology of~\cite{Otto2004}, ``the equivalence relation $\equiv_{\unfo^k_{l}}$ can be lifted to
$\equiv_{\fo^q}$ modulo $\unkbisim$''.  This implies the theorem:
starting with a $\unkbisim$-invariant \fo formula $\phi$ in one free variable of quantifier depth
$q$, we obtain that $\phi$ is {$\equiv_{\unfo^k_{l}}$-invariant}. Therefore $\phi$ is equivalent to the disjunction of
all (finitely many, up to equivalence) formulas describing an equivalence class of $\equiv_{\unfo^k_{l}}$ containing a model
of $\phi$.

The construction of $(M^*,a^*)$ and $(N^*,b^*)$ makes use of
Theorem~\ref{thm-otto-cover}. We first need to introduce some auxiliary
definitions that will allow us to construct a Kripke structure from an arbitrary structure, and
vice versa.

Let $M$ be any structure over a signature $\sigma$. In what follows, by a
\emph{$k$-neighborhood} of $M$, we will mean a pair $(K,h)$, where $K$ is
a structure with domain $\set{1, \ldots, k}$ over $\sigma$ and $h$ is a
homomorphism from $K$ to $M$. Intuitively, one can think of a
$k$-neighborhood as a realization in $M$ of some positive existential
description of size $k$. 
From $\sigma$ and $k$ we construct the signature
$\sigma_k$ containing a unary predicate $P_K$ for each structure $K$ over
$\sigma$ with domain $\set{1, \ldots, k}$ and a binary relation
$R_i$ for each $1\leq i\leq
k$. Given a structure $M$ over $\sigma$, we associate to it a
structure $G_M$, called \emph{the graph of
  $M$}, over $\sigma_k$, as follows.  The nodes of $G_M$ are the elements of
$M$ plus the $k$-neighborhoods of $M$. Each node of $G_M$ that is a $k$-neighborhood $(K,h)$ is labeled
by the unary predicate $P_K$.  There is an $R_i$-edge in $G_M$ from a $k$-neighborhood $(K,h)$ of $M$ to an element $a\in M$ if
and only if $h(i)=a$.

Conversely, we can transform a structure $G$ over the signature $\sigma_k$
into a structure $\widehat{G}$ over $\sigma$ as follows. The universe
of $\widehat G$ consists of the nodes of $G$ that do not satisfy
any unary predicate. A tuple of such elements $\tuple{a}=a_1\ldots
a_n$ (with $n\leq m\leq k$) belongs to
$R$ 
iff there is a node $u$ of $G$ satisfying $P_K$, for some $K$, 
   and a tuple $\tuple{b}=b_1\ldots b_n$ of elements of
  $K$ such that $K\models R(\tuple{b})$ and such that
   for each $i\leq n$, $G$ contains the edge $R_{b_i}(u,a_i)$.

We start by establishing a useful property linking $M$ to its graph $G_M$.

  \begin{claim}\label{claim-graph-struct-formula}
For all structures $(M,a)$ and $(N,b)$ and  for all $l\in\mathbb{N}$, if   $(M,a)\equiv_{\unfo^k_{l}} (N,b)$ then   $(G_M,a)\bisim^l (G_N,b)$.
  \end{claim}

  \begin{proof}
    Let $Z_0$ be the set of all pairs of elements $(a,b)$ such that $(M,a)
    \equiv_{\unfo^k_{0}} (N,b)$, and, for $\ell>0$, let $Z_{\ell}$ be the set
    of all pairs of elements $(a,b)$ such that $(M,a) \equiv_{\unfo^k_{\ell}}
    (N,b)$, together with the set of all pairs $((K,h),(K,g))$ such that, for
    all $i\leq k$, $(M,h(i)) \equiv_{\unfo^k_{\ell-1}} (N, g(i))$.  We show, by
    induction on $\ell$, that from a pair $(u,v)$ in $Z_{\ell}$, Elo\"ise can
    survive $\ell$ steps of the global two-way bisimulation game. Therefore, if
    $(a,b)$ is such that $(M,a)\equiv_{\unfo^k_{l}} (N,b)$, $Z_l$ witnesses the
    fact that $(G_M,a)\bisim^l (G_N,b)$ and the lemma is proved.

The base of the induction, where $\ell=0$, is trivial. For the
induction step,
recall that Abelard can choose a structure, and subsequently, play one
of three types of moves: moving the pebble forward along an edge,
moving the pebble backward along an edge, and moving the pebble to an
arbitrary node. We will assume that Abelard chooses $G_M$ (the
case where Abelard chooses $G_N$ is symmetric) and that Abelard either
(i) moves the pebble (forward or backward) along an edge, 
or (ii) moves the pebble to an arbitrary node. We treat the two cases
separately. 

As a convenient notation, for each structure $K$ over $\sigma$ with universe
$\set{1,\ldots,k}$, we define $\alpha_K(x_1,\ldots,x_k)$ to be the
quantifier-free formula that is the conjunction of all atomic formulas true in
$K$, using free variables $x_1, \ldots, x_k$ for the elements $1, \ldots, k$.
Moreover we say that $\tau$ is the $\unfo^k_{\ell}$-type of an element $a$ of
$M$ if it is the conjunction of all the $\unfo^k_{\ell}$ formula $\varphi(x)$
such that $M\models \varphi(a)$. Modulo logical equivalence there are only
finitely many formulas in $\unfo^k_{\ell}$ hence each $\unfo^k_{\ell}$-type is
equivalent to a formula of $\unfo^k_{\ell}$.

(i) Let $(u,v)\in Z_\ell$. By construction, $(u,v)$ is either of the
form  $(a,b)$ with $a\in \dom(M)$ and $b\in \dom(N)$, or it is 
of the form $((K,h),(K,g))$. First, 
suppose that $(u,v)=(a,b)$, where $a\in M$ and $b\in N$. Let $u'$  be a
node of $G_M$ reachable from $a$ by an edge. By construction of $G_M$, $u'$
must be of the form $(K,h)$ and is connected with an edge of label $i$ to $a$,
i.e. $h(i)=a$. For all $i$, let $\tau _i$ be the $\unfo^k_{\ell-1}$-type of $h(i)$. By
definition, the formula $\beta(x_i)$ defined as $\exists x_1 \ldots x_{i-1}
x_{i+1} \ldots x_k ~\alpha_K \land \bigwedge_{j\neq i} \tau_j(x_j)$ is a formula
of $\unfo^k_{\ell}$ and by construction we have $(M,a)\models
\beta(x_i)$. Therefore $(N,b)\models \beta(x_i)$. Let $g$ be the valuation
corresponding to the initial block of existential quantifications. By
construction $g$ is a homomorphism from $K$ to $N$ and, for each $i$,
$(h(i),g(i))\in Z_{\ell-1}$. Hence $((K,h),(K,g))\in Z_{\ell-1}$. 

Next, suppose $(u,v)=((K,h),(K,g))$ and a there is a node $u'$ connected
to $u$ via an edge. By construction $u'$ must be an element $a$ of $G_M$
connected with an edge of label $i$ to $(K,h)$, i.e. $h(i)=a$. By definition,
for $b=g(i)$ we have $(a,b) \in Z_{\ell-1}$.

(ii) Let $(u,v)\in Z_\ell$ and let $u'$ be any element of
$G_M$. First, consider the case where $u'$ is an element of $M$, and
let $\tau(x)$ be its $\unfo^k_{\ell-1}$-type. Since
$(M,a)\equiv_{\unfo^k_\ell} (N,b)$ and $(M,a)\models \exists x~\tau(x)$, 
we have that $(N,b)\models\exists x~\tau(x)$. Let $v'$ be the
witnessing
 element of $N$. Then $(N,u)\equiv_{\unfo^k_{\ell-1}} (N,v)$, and hence,
$(u',v')\in Z_{\ell-1}$. 

Next, suppose that $u'$ was of the form $(K,h)$. For each $i\leq k$,
let $\tau_i(x)$ be the $\unfo^k_{\ell-1}$-type of $h(i)$. By
definition, the formula $\beta$ defined as $\exists x_1 \ldots x_k
~\alpha_K \land \bigwedge_{i} \tau_i(x_i)$ is a formula
of $\unfo^k_{\ell}$ and by construction we have $(M,a)\models
\beta$. Therefore $(N,b)\models \beta$. Let $g$ be the valuation
corresponding to the initial block of existential quantifications in $\beta$. By
construction $g$ is a homomorphism from $K$ to $N$ and, for each $i$,
$(h(i),g(i))\in Z_{\ell-1}$. Hence $((K,h),(K,g))\in Z_{\ell-1}$. 
\end{proof}

Fix some $q\in\mathbb{N}$ and let $l$ be the number given by
Theorem~\ref{thm-otto-cover} for $q+1$.  Let now $(M,a)$ and $(N,b)$ be such that
$(M,a)\equiv_{\unfo^k_{l}} (N,b)$. By Claim~\ref{claim-graph-struct-formula} we have
$(G_M,a)\bisim^l (G_N,b)$.  By Theorem~\ref{thm-otto-cover} there exists $(G'_M,a')$ and
$(G'_N,b')$ such that $(G'_M,a') \bisim (G_M,a)$, $(G'_N,b') \bisim (G_N,b)$ and $(G'_M,a)
\equiv_{\fo^{q+1}} (G'_N,b)$.
Let $M^*$ be $\widehat G'_M$ and $N^*$ be $\widehat G'_N$. The following two
claims show that $(M^*,a')$ and $(N^*,b')$ have the desired properties.

  \begin{claim}\label{claim2}
  $(M^*,a') \equiv_{\fo^q} (N^*,b')$.
  \end{claim}

  \begin{proof}
    We  lift the winning strategy in the Ehrenfeucht-Fra\"iss\'e game given by $(G'_M,a') \equiv_{\fo^{q+1}} (G'_N,b')$
    to a winning strategy for the Ehrenfeucht-Fra\"iss\'e game
    $(M^*,a') \equiv_{\fo^q} (N^*,b')$: When, at stage $i$ of the game,
    Spoiler chooses an element $a_i$ in $M^*$ Duplicator responds with the
    element $b_i$ provided by his strategy in the game between $G'_M$ and
    $G'_N$, and vice versa if Spoiler chooses an element in $N^*$. We 
    show that $M^*\models R(a_{i_1} \cdots a_{i_\kappa})$ implies $N^*\models R(b_{i_1}
    \cdots b_{i_\kappa})$. Let $(K,h)$ be any $k$-neighborhood of $M$ such that
    $R(a_{i_1},\ldots, a_{i_\kappa})$ is the $h$-image of an atomic
    fact of $K$. Assuming the
    node $u=(K,h)$ is played in the game between $G'_M$ and $G'_N$,
    the extra move in the winning  strategy of Duplicator guarantees the
    existence of $v=(K,g)$ in $G'_M$ that is linked to $\tuple{b}$ the same way
    $u$ is linked to $\tuple{a}$. This implies the desired property.
  \end{proof}

  \begin{claim}\label{claim3}
  $(M^*,a')\unkbisim (M,a)$ and $(N^*,b') \unkbisim (N,b)$.  
  \end{claim}
  \begin{proof}
    We only treat the case for $M$, the other being identical.  

    By
    Theorem~\ref{thm-otto-cover}, we have that $(G'_M,a')$ is a 
    $\bisim$-cover of $(G_M,a)$. That is, there is a homomorphism  
    $h:G'_M\to G_M$ such that $h(a')=a$ and such that 
    $(G'_M,c)\bisim (G_M,h(c))$ for all $c\in \dom(G'_M)$. 

   Let $Z=\{(c,h(c))\mid c\in \dom(M^*)\}$. Note that, for all
   $c\in \dom(M^*)$ we have that $h(c)\in \dom(M)$, due to 
   the fact that $c$ and $h(c)$ agree on all unary predicates.
   We will show that, in fact, $Z$ is a UN-bisimulation of
  width $k$ between $M^*$ and $M$. Note also that $(a',a)\in Z$. 

 We first prove that $Z$ is the graph of a homomorphism. 
  Let $\tuple{a}$ be such that $M^*
 \models R(\tuple{a})$.
By construction of $M^*$, this implies that there is a node $u$ in $G'_M$ of
label $P_K$ and
elements $b_1,\ldots,b_l$ in $K$ such that $K\models R(b_1,\ldots,b_l)$ and
$G'_M \models R_{b_j}(u,a_j)$ for all $j$.
As $h$ is a homomorphism, $h(u)$ has the same label as $u$ and is connected to
$h(\tuple{a})$ the same way $u$ is connected to $\tuple{a}$. By construction of
$G_M$, this implies $M\models R(h(\tuple{a}))$.

Since we have just shown that $Z$ is the graph of a homomorphism, it remains only to
show that $Z$ satisfies the backward property of the definition of
UN-bisimulations.

Consider $(c,d)\in Z$, that is, $d=h(c)$. Let $Y$ be a subset of elements of
$M$ of size less than $k$. We assume that $d\in Y$. The case where $d\not\in Y$
is handled in the same way. Let $K$ be the structure with universe $\set{1,\ldots,k}$
which is isomorphic to the restriction of $M$ to $Y$ and let $g$ be the
corresponding isomorphism. Without loss of generality we assume that $g(1)=b$.
By definition of $G_M$ there is a node $v=(K,g)$ that is connected to all the
elements of $Y$ via edges of appropriate label. Since $(G'_M,c)\bisim (G_M,d)$, we
can find a node $u$, whose label correspond to $K$, connected to $c$ via an
edge of label $1$, and such that $h(u)=v$.  Since $(G'_M,u)\bisim (G_M,v)$, we can
further find in $G'_M$ a set of nodes $X$, connected to $u$ in the same way
that $Y$ are connected to $v$, such that $h(X)=Y$.  It follows by the
construction of $M^*$ that the map from $Y$ to $X$ is a partial homomorphism
from $M$ to $M^*$.
  \end{proof}

  This concludes the proof of Theorem~\ref{thm-charac-fo-finite}: it
  follows from Claim~\ref{claim2} and Claim~\ref{claim3} that every
  $\unkbisim$-invariant \fo formula $\phi$ in one free variable of
  quantifier depth $q$, is {$\equiv_{\unfo^k_{l}}$-invariant} for
  suitably large $l$, and, therefore (as explained earlier), is
  equivalent to a $\unfo^k_l$-formula. The analogous result for
  sentences follows immediately.
\end{proof}

\subsection{Craig Interpolation and Beth Property}

We conclude the list of nice model-theoretic properties of \unfo by showing
that it has Craig Interpolation and the Projective Beth Property. In
fact, we can show strong versions of these results, which take into account also
the width of formulas. This is remarkable, given that both Craig Interpolation
and the Beth Property fail for the $k$-variable fragment of first-order logic,
for all $k>1$. Moreover, the results presented in this section hold both on
arbitrary structures and on finite structures. Model-theoretic proofs
of Craig Interpolation typically use amalgamation
constructions \cite{ChangKeisler}, and the proof we give here
is essentially based on an amalgamation construction called 
\emph{zigzag-products} that was introduced, in the context of
modal logic, in~\cite{MarxVenema}.

For all \unfo-formulas
$\phi(\tuple{x}), \psi(\tuple{x})$, we write $\phi\models\psi$ to
express that the first-order formula $\forall\tuple{x}(\phi\to\psi)$
(which is not necessarily a \unfo-formula) is valid, i.e. holds on all models.

It turns out that it is enough to consider finite models for testing whether
$\phi\models\psi$. This is a consequence of the following remark and
Theorem~\ref{thm-finite-model-prop}:

\begin{remark}\label{remark-entailment} For all $\unfo$-formulas
  $\phi(\tuple{x}),\psi(\tuple{x})$, $\phi \models \psi$ holds (resp. holds on
  finite structures) if and only if the \unfo sentence
  \[ \exists \tuple{x}(\phi\land\bigwedge_{1\leq i\leq n} P_{i}(x_i))\land\neg\exists
  \tuple{x}(\psi\land\bigwedge_{1\leq i\leq n} P_{i}(x_i)) \] is not satisfiable (resp. does
  not have a finite model), where $\tuple{x}=x_1, \ldots, x_n$ and
  $P_1, \ldots, P_n$ are fresh unary predicates.
\end{remark} 

Remark~\ref{remark-entailment} implies that all the results we prove
for sentences apply to entailment between formulas with free variables
as well. In particular, since \unfo has the finite model property by
Theorem~\ref{thm-finite-model-prop}, we have that $\phi \models \psi$
holds on arbitrary structures if and only if holds on finite
structures. Consequently, Craig Interpolation, and therefore also Beth
Property, hold on arbitrary structures if and only if they
hold over finite structures. In the remaining part of this section we
only state the results for arbitrary structures, but, as we have just argued,
they also hold over finite structures.

\begin{theorem}\label{thm-craig}
  $\unfo^k$ has Craig Interpolation: for all $k\geq 1$ and for every pair of
  $\unfo^k$-formulas $\phi(\tuple{x}), \psi(\tuple{x})$ in the same free
  variables such that $\phi\models\psi$, there is a $\unfo^k$-formula
  $\chi(\tuple{x})$ over the common vocabulary of $\phi$ and $\psi$ such that
  $\phi\models\chi$ and $\chi\models\psi$.
\end{theorem}

\begin{proof}
  The proof is by contradiction. Suppose that there is no interpolant
  in $\unfo^k$. We will show that (the first-order formula)
  $\phi\land\neg\psi$ is satisfiable. 

  Let $\sigma$ be the vocabulary of $\phi$, and $\tau$ the vocabulary
  of $\psi$, and let $\tuple{x}=x_1\ldots x_n$.  
  By $(M,\tuple{a})\ \equiv_{\unfo^k}^\sigma (N,\tuple{b})$ we will
    mean that the tuple $\tuple{a}$ in $M$ and the tuple $\tuple{b}$
    in $N$ satisfy the same $\unfo^k$-formulas using only relation symbols
    in $\sigma$, and by $(M,\tuple{a})\ \Rrightarrow_{\unfo^k}^\sigma (N,\tuple{b})$ we will
    mean that every $\unfo^k$-formula using only relation symbols
    in $\sigma$ that is satisfied by the tuple $\tuple{a}$
    of $M$ is satisfied also by the tuple $\tuple{b}$ of $N$.

\begin{claim}\label{claim-one}
 There are $(M,\tuple{a})\models\phi$ and $(N,\tuple{b})\models\neg\psi$
such that $(M,\tuple{a})\Rrightarrow_{\unfo^k}^{\sigma\cap\tau} (N,\tuple{b})$.
\end{claim}

\begin{proof}[Proof of Claim~\ref{claim-one}:] 
  A classic argument involving two applications of Compactness
  (cf.~\cite[Lemma 3.2.1]{ChangKeisler}).
  Let $\Phi(\tuple{x})$ be the set of all $\unfo^k$-formulas in the
  joint vocabulary that are valid consequences of $\phi(\tuple{x})$.
  The first-order theory $\Phi(\tuple{x})\cup
  \{\neg\psi(\tuple{x})\}$ is consistent, because, if this were not the
  case, then by Compactness, some finite conjunction of formulas in
  $\Phi$ would be an interpolant, which we have assumed is not the case. Hence, let
  $(N,\tuple{b})\models\Phi(\tuple{x})\cup\{\neg\psi(\tuple{x})\}$. 
  
Next, let
$\Gamma(\tuple{x})$ be the set of all $\unfo^k$-formulas in the joint vocabulary
that are \emph{false} in $(N,\tuple{b})$. Then the first-order theory
$\{\neg\gamma\mid\gamma\in\Gamma(\tuple{x})\}\cup\{\phi(\tuple{x})\}$ is consistent, for, if it were not, then by
Compactness, there would be $\gamma_1, \ldots, \gamma_m\in\Gamma$ such
that $\bigwedge_{i}\neg\gamma_i\models\neg\phi(\tuple{x})$ and hence
$\bigvee_{i}\gamma_i\in \Phi(\tuple{x})$, which contradicts $(N,\tuple{b})\models\Phi(\tuple{x})$.
Let $(M,\tuple{a})\models
\{\neg\gamma\mid\gamma\in\Gamma(\tuple{x})\}\cup\{\phi(\tuple{x})\}$. 
By construction, all $\unfo^k$ formulas in $\sigma\cap\tau$ true of
$\tuple{a}$ in $M$ are true of $\tuple{b}$ in $N$. In other words,
$(M,\tuple{a})\Rrightarrow_{\unfo^k}^{\sigma\cap\tau} (N,\tuple{b})$. 
\end{proof}

The following Claim~\ref{claim:unhom} is a strengthening of
Claim~\ref{claim-one}.  Recall the definition of UN-homomorphisms. 
Just as we parametrized UN-bisimulations by a width $k$ and a signature, we can 
parametrize UN-homomorphisms by a width and a signature.
We write $(M,x) ~\approx_{\text{UN}^k}^{\sigma\cap\tau}~ (N,y)$ if there is a
UN-bisimulation of width~$k$, with respect to the signature $\sigma\cap\tau$,
between $(M,x)$ and $(N,y)$.
We write $(M,\tuple{a})
~\to_{\text{UN}^k}^{\sigma\cap\tau}~ (N,\tuple{b})$ if, for every set
$X\subseteq M$ with $|X|\leq k$, there is a partial homomorphism (with
respect to the signature $\sigma\cap\tau$)
$h:M\to N$ with domain $X$, such that (i) $h(a_i)=b_i$ for all $a_i\in \tuple{a}\cap X$,
and (ii) $(M,x) ~\approx_{\text{UN}^k}^{\sigma\cap\tau}~ (N,h(x))$ for all $x\in
X$.  In particular this implies that $(M,\tuple{a})\Rrightarrow_{\unfo^k}^{\sigma\cap\tau} (N,\tuple{b})$. 

\begin{claim}\label{claim:unhom}
 There are $(M,\tuple{a})\models\phi$ and $(N,\tuple{b})\models\neg\psi$
such that $(M,\tuple{a})~\to_{\text{UN}^k}^{\sigma\cap\tau}~
(N,\tuple{b})$.
\end{claim}

\begin{proof}[Proof of Claim~\ref{claim:unhom}:]
We may assume without loss of generality that the models $M$ and $N$
provided by Claim~\ref{claim-one} are $\omega$-saturated, and
therefore, by Lemma~\ref{lem:saturated} we have that, whenever
$(M,a) ~\equiv_{\unfo^k}^{\sigma\cap\tau}~(N,b)$, then 
$(M,a) ~\approx_{\text{UN}^k}^{\sigma\cap\tau}~(N,b)$.

Consider now a set $X=\set{u_1,\ldots,u_\kappa} \subseteq M$ with $\kappa\leq
k$.  We assume that $X$ contains no element of $\tuple{a}$. If this were not
the case we simply remove those elements and proceed as below. Let $Y=X \cup
\tuple{a}$. We view $Y$ as the sequence $(v_1,\ldots,v_\ell)$ where for $j\leq
\kappa$, $v_j=u_j$ and for $j>\kappa$, $v_j = a_j$. For each element
$u_i$ of $X$, let $T_i$ be the set of $\unfo^k$ formulas $\psi(x)$ such that $M
\models \psi(u_i)$. Let $\alpha$ be the conjunction of all atoms
$R(x_{i_1},\ldots,x_{i_l})$ such that $M\models R(v_{i_1},\ldots,v_{i_l})$.
For each finite subset $T'_i$ of $T_i$, the formula $\exists x_1,\ldots
x_\kappa ~\alpha \land \bigwedge_i T'_i(x_i)$ is in $\unfo^k$ and is satisfied by
$(M,\tuple{a})$. Because $(M,\tuple{a})
~\Rrightarrow_{\unfo^k}^{\sigma\cap\tau}~(N,\tuple{b})$ it is also satisfied by
$(N,\tuple{b})$. By $\omega$-saturation, this implies that 
$\{\alpha(x_1, \ldots, x_i)\}\cup \bigcup_{i\leq \kappa} T_i(x_i)$ is
realized in
$(N,\tuple{b})$. The witnessing tuple $\tuple{c}$ provides the desired homomorphism from $X$ to $N$.

Altogether this shows that $(M,\tuple{a})
~\to_{\text{UN}^k}^{\sigma\cap\tau}~(N,\tuple{b})$.
\end{proof}

\noindent We now construct a new $\sigma\cup\tau$-structure $K$ out of the
$\sigma$-structure $M$ and the $\tau$-structure $N$ provided by Claim~\ref{claim:unhom}.  This new
structure $K$ will contain a tuple satisfying $\phi\land\neg\psi$. 
Essentially, $K$ will be the substructure of the cartesian product of
$M$ and $N$ containing pairs of elements that are UN-bisimilar in the
joint language. Relations in the joint vocabulary are interpreted as
usual in cartesian products, while relations outside of the joint
vocabulary are copied from the respective structure.
The precise definition of $K$ is as follows:

\begin{itemize}
\item The domain of $K$ is the set of all pairs
$(a,b)\in M\times N$ such that $(M,a) ~\approx_{\text{UN}^k}^{\sigma\cap\tau}~ (N,b)$.
\item For each $R\in\sigma\cap\tau$ of arity $m$, $R^{K}$
 contains all tuples $(\langle a_1,b_1\rangle,\ldots,\langle
  a_m,b_m\rangle)$ in the domain of $K$ such that $(a_1, \ldots, a_m)\in R^M$ and $(b_1,
  \ldots, b_m)\in R^N$. 
\item For each $R\in \sigma\setminus \tau$ of arity $m$
  $R^{K}$ 
   contains all tuples $(\langle a_1,b_1\rangle,\ldots,\langle
  a_m,b_m\rangle)$ in the domain of $K$ such that $(a_1, \ldots, a_m)\in R^M$.
\item Analogously for all relation symbols $R\in\tau\setminus\sigma$. 
\end{itemize}

\begin{claim}\label{claim-product} 
  For all elements $\langle a,b\rangle$ of $K$, we have 
\[(K,\langle
  a,b\rangle) ~\approx_{\text{UN}^k}^\sigma ~ (M,a)
  \quad\hbox{and}\quad(K,\langle
  a,b\rangle) ~\approx_{\text{UN}^k}^\tau~ (N,b)\ .
\] 
\end{claim}

\begin{proof}[Proof of Claim~\ref{claim-product}:]
We will show that $(K,\langle
  a,b\rangle) ~\approx_{\text{UN}^k}^\sigma~ (M,a)$. The proof of the other half
of the claim is analogous.

Let $Z$ be the graph of the natural projection from $K$ onto $M$,
i.e., $Z = \{(\langle u,v\rangle,u)\mid \langle u,v\rangle\in K\}$. 
Clearly, $(\langle a,b\rangle,a)\in Z$. Therefore, it suffices to show 
  that $Z$ is a UN-bisimulation of width $k$ for the signature
  $\sigma$. 

  Consider any pair $(\langle u,v\rangle,u)\in Z$, and let $X$ be any
  subset of the domain of $K$, with $|X|\leq k$.  Let $h$ be the
  natural projection from $K$ onto $M$, restricted to $X$. It is clear
  from the definition of $K$ that $h$ is a partial homomorphism with
  respect to all relations in $\sigma$ (those that belong to
  $\sigma\cap\tau$ as well as those that belong to
  $\sigma\setminus\tau$). Furthermore, it is clear that $h(\langle
  a,b\rangle)=a$ if $\langle a,b\rangle\in X$.  Therefore, the forward property
  of the definition of UN-bisimulations is satisfied.

  Next, consider any  pair $(\langle u,v\rangle,u)\in Z$, and let $X$ be any
  subset of the domain of $M$, with $|X|\leq k$. Since
  $(M,u)~\approx_{\text{UN}^{k}}^{\sigma\cap\tau}~(N,v)$, there is a 
  partial homomorphism $f:M\to N$ whose domain is $X$, such that
  $(M,x)~\approx_{\text{UN}^{k}}^{\sigma\cap\tau}~(N,h(x))$ for all $x\in X$, and
  such that $h(u)=v$ if $u\in X$. Now, define $h:M\to K$ to be the 
  map that sends every $x\in X$ to $\langle x,f(x)\rangle$. Clearly, 
  this is a well-defined partial map from $M$ to $K$, with domain $X$,
  having the property that $h(u)=(u,v)$ if $u\in X$. Therefore, it
  only
  remains to show that $h$ is a partial \emph{homomorphism}. 
  That $h$ preserves all relations in $\sigma\setminus\tau$ is immediate
  from the definition of $K$. That $h$ preserves all relations in
  $\sigma\cap\tau$ follows from the construction of $h$ and of $K$: if 
  $(x_1, \ldots, x_m)\in R^M$, then, since $f$ is a partial
  homomorphism,  we have that $(f(x_1), \ldots, f(x_m))\in R^N$. Hence, 
  by the definition of $K$, we have that $(h(x_1), \ldots, h(x_m))\in
  R^K$. 
\end{proof}

We will use the notation $\langle\tuple{a},\tuple{b}\rangle$ 
as a convenient shorthand for $(\langle a_1,b_1\rangle, \ldots, \langle
a_n,b_n\rangle)$. 

\begin{claim}\label{claim:phi-not-psi}
$(K,\langle\tuple{a},\tuple{b}\rangle)\models\phi\land\neg\psi$ 
\end{claim}

\begin{proof}[Proof of Claim~\ref{claim:phi-not-psi}:] We may assume that $\phi$
  is in \unnormalform.  Let $\phi$ be of the form $\exists\tuple{z}\phi'$,
  where $\phi'$ is built up from atomic formulas, equalities and formulas in
  one free variable, using conjunction and disjunction.  Since
  $(M,\tuple{a})\models\phi$, there is a tuple $\tuple{a}'$ witnessing
  $(M,\tuple{a}')\models\phi'$ (where $\tuple{a}$ and $\tuple{a}'$ are
  appropriately related, depending on which variables are quantified in
  $\phi$).  The fact that $(M,\tuple{a})~\to_{\text{UN}^k}^{\sigma\cap\tau}
  (N,\tuple{b})$ gives us that there is a partial homomorphism $h:M\to N$ whose
  domain is $\tuple{a}'$, such that $h(a_i)=b_i$ for all $a_i\in \tuple{a}\cap
  \tuple{a}'$, and such that $(M,x)~\approx_{\text{UN}^k}^{\sigma\cap\tau}
  (N,h(x))$ for all $x\in \tuple{a}'$.  Now, let us use
  $\langle\tuple{a}',\tuple{b}'\rangle$ again as a convenient shorthand for the
  pairwise product of $\tuple{a}'$ and $\tuple{b}'$. Then it follows from the
  definition of $K$ that all atomic formulas in signature $\tau$ that are true
  of $\tuple{a}'$ in $M$ are true of $\langle\tuple{a}',\tuple{b}'\rangle$ in
  $K$ (for relations in $\tau\setminus\sigma$, this is immediate, and for
  relations in $\sigma\cap\tau$, this follows from the fact that the atomic
  formula in question is true also of $\tuple{b}'$ in $N$). Hence (by induction
  on the conjunctions and disjunction in $\phi'$), we have that
  $(K,\langle\tuple{a}',\tuple{b}'\rangle)\models\phi'$. In fact, it can be
  seen that $\langle\tuple{a}',\tuple{b}'\rangle$ provides a witness showing
  that $(K,\langle\tuple{a},\tuple{b}\rangle)\models\phi$.

Next, consider $\psi$. We may again assume that $\psi$ is in
\unnormalform.
Recall that the definition of $K$ implies that
the natural projections from $K$ onto $N$ is a homomorphism
preserving all relations in $\tau$ (both those in $\sigma\cap\tau$ and
those in $\tau\setminus\sigma$).  It follows by  a straightforward
induction that if
$(K,\langle\tuple{a},\tuple{b}\rangle)\models\psi$, then also
$(N,\tuple{a})\models\psi$, which is not the case.
\end{proof}

To summarize: on the basis of the assumption that there is no Craig
interpolant for $\phi$ and $\psi$, we were able to show that 
$\phi\land\neg\psi$ is satisfiable, and hence, $\phi\to\psi$ is not a
valid implication. This concludes the proof of our Craig interpolation theorem.
\end{proof}

As usual, Craig Interpolation implies Beth Property.
Let $\Sigma$ be a \unfo-theory
in a signature $\sigma$ and let $R\in\sigma$ and
$\tau\subseteq\sigma$.  We say that $\Sigma$ \emph{implicitly defines}
$R$ in terms of $\tau$ if for all $\tau$-structures $M$ and for all
$\sigma$-expansions $M_1, M_2$ of $M$ satisfying $\Sigma$, we have
that $R^{M_1} = R^{M_2}$.  We say that a formula $\phi(\tuple{x})$ in
signature $\tau$ is an \emph{explicit definition} of $R$ relative to
$\Sigma$ if
$\Sigma\models\forall\tuple{x}~~(R\tuple{x}\leftrightarrow\phi(\tuple{x}))$. 
Note that the formula $\forall\tuple{x}~~(R\tuple{x}
  \leftrightarrow \phi(\tuple{x}))$ is itself not necessarily a \unfo-formula,
but this is irrelevant. 

\begin{theorem}\label{thm-beth}
  \unfo has the Projective Beth Property: whenever a \unfo-theory
  $\Sigma$ in a signature $\sigma$ implicitly defines a $k$-relation
  $R$ in terms of a signature $\tau\subseteq\sigma$, then there is a
  \unfo-formula in signature $\tau$ that is an explicit definition of
  $R$ relative to $\Sigma$.
  Moreover, if $\Sigma$ belongs to $\unfo^k$ ($k\geq 1$), then the explicit
  definition  can be found in $\unfo^k$ as well.
\end{theorem}

\begin{proof}
  The argument is standard from Craig Interpolation. We give
  it here for the sake of completeness.  Suppose $\Sigma$ implicitly defines
  $R$ in terms of $\tau$.  We may assume $R\not\in\tau$ because otherwise it is
  trivial.  Furthermore, by Compactness, we may assume $\Sigma$ to be finite.
  Let $\Sigma'$ be a copy of $\Sigma$ in which all relations $S$ outside of
  $\tau$ (including $R$) have been replaced by new disjoint copies $S'$. Then
  the fact that $\Sigma$ implicitly defines $R$ in terms of $\tau$ implies that
  following first-order implication is valid:

$$\bigwedge\Sigma \land \bigwedge\Sigma' \models \forall
\tuple{x}(R\tuple{x}\to R'\tuple{x})$$

We can rewrite this into an equivalent implication of \unfo-formulas:

$$\bigwedge\Sigma \land R\tuple{x} \models \bigwedge\Sigma'\to R\tuple{x}$$

Let $\chi(\tuple{x})$ be the interpolant given by Theorem~\ref{thm-craig} for this \unfo-implication. It
is now straightforward to verify that $\chi$ is indeed an explicit definition of
$R$ relative to $\Sigma$. 
\end{proof}

\section{Satisfiability}\label{section-sat}

In this section, we show that the satisfiability problem for \unfp and for
\unfo is {\twoexptime-complete}, both on arbitrary structures and on finite
structures.  The lower bound holds already over structures with
relations of bounded arity, and, in particular, over finite trees.
Note that this is in contrast with \gfo whose complexity drops from
\twoexptime-complete to \exptime-complete when the arity of relations is
bounded~\cite{Gradel01}.  The upper bound is obtained by a reduction to the
two-way modal \mucalc and Theorem~\ref{thm-decid-mu}. Given a formula $\varphi$
of \unfp we construct in exponential time a formula $\varphi^*$ in the \mucalc
such that $\varphi$ has a (finite) model iff $\varphi^*$ has a (finite)
model. The correctness of the construction in the finite case builds on
Theorem~\ref{thm-otto-acyclic-cover}.  The same reduction to the two-way modal
\mucalc allows us to prove the finite model property of \unfo and the tree-like
model property of \unfp, i.e., Theorem~\ref{thm-finite-model-prop} and
Theorem~\ref{thm-tree-like-prop}.

We describe the reduction from $\varphi$ to $\varphi^*$ in two parts. In the
first one we consider only a special case of \unfp formulas that we call
\emph{simple}. Those are, intuitively, formulas of the global two-way \mucalc
with navigation through arbitrary relations instead of just binary relations.
The construction of $\varphi^*$ is then polynomial. In a second part we show
how the general case reduces to this one (with an exponential blow-up).

\subsection{Simple \unfp formulas}
We first consider a fragment of \unfp, which we call \emph{simple} and denote it by
\sunfp. It is a common fragment of \unfp and \gfp, which embeds the global two-way
$\mu$-calculus.  The syntax of \sunfp is given by the following grammar (recall
that we use the notation $\phi(x)$ to indicate that a formula has no
free first-order variables besides possibly $x$, but may contain some monadic
second-order free variables):
\begin{align*}
  \phi(x)::=~~~ & P(x) ~|~ X(x) ~|~ \phi(x) \wedge \phi(x) ~|~
  \phi(x) \vee \phi(x) ~|~ \neg \phi(x) 
  ~|~ [\LFP_{X,y}\phi(y)](x) ~|~ 
\\ &
  \exists y_1\ldots y_n(R(y_1\ldots y_n)\land y_i=x \land
  \!\!\!\!\!\!\!\!\bigwedge_{j\in\{1\ldots n\}\setminus\{i\}}\!\!\!\!\!\!\!\!\phi(y_j)) ~\mid~ \exists x \phi(x)
\end{align*}

Note that all formulas generated by this inductive definition have at most one
free (first-order) variable.  We denote by \sunfo the first-order
(that is, fixpoint-free) fragment
of \sunfp.

We need the following notions. A \emph{fact} of a structure $M$ is an
expression $R(a_1, \ldots, a_n)$ where $(a_1, \ldots, a_n)\in R^M$.
The \emph{incidence graph} $\inc(M)$ of a
structure $M$ is the bi-partite graph containing facts of $M$ and elements of
$M$, and with an edge between a fact and an element if the element occurs in
the fact. We say that a structure $M$ is $l$-acyclic, for $l\geq 1$,
if (i)
$\inc(M)$ has no cycle of length less than $2l$, and (ii) no element of $M$ occurs
twice in the same fact.  We call a structure acyclic, if it is $l$-acyclic for
all $l$ (i.e., the incidence graph is acyclic and no element occurs in the same
fact twice).

Based on a simple coding of relations of arbitrary arity using binary relations
we can transform a simple formula into a formula of the \mucalc and, using the
results of Section~\ref{sec-modal-logic}, obtain:

\begin{prop}\label{prop-sunfp}~
\begin{enumerate}
\item
The satisfiability problem for \sunfp is \exptime-complete, both 
on arbitrary structures and on finite structures. 
\item
If a \sunfp formula has a model, it has an acyclic model
\item
If a \sunfp formula  has a finite model, then it has a $l$-acyclic
finite model, $\forall l\geq 1$.
\item 
\sunfo has the finite model property.
\end{enumerate}
\end{prop}

\begin{proof}
  (1) 
  The \exptime lower bound for \sunfp follows immediately from the fact that
  \sunfp subsumes the two-way modal $\mu$-calculus $\ml_\mu$ whose
  satisfiability problem
  is \exptime-hard (cf.~Theorem~\ref{thm-decid-mu}).  For the \exptime upper
  bounds, we give a polynomial time translation from \sunfp to $\ml_\mu$, which
  preserves (un)satisfiability on arbitrary structures and on finite
  structures.

  Let $\phi$ be any \sunfp formula. Let $l$ be
  the maximal arity of the relation symbols occurring in $\phi$, and let $p_1,
  \ldots, p_l$ be fresh proposition letters. For each unary predicate $P$, we
  introduce a fresh proposition letter $p_P$. Furthermore, we associate to each
  relation symbol $R$ a proposition letter $p_R$. Intuitively, the
  idea of the translation is that, in our Kripke models,  we will encode  an $R$-tuple by
  a ``gadget'' consisting of a node satisfying $p_R$ that
  has $n$ successors (where $n$ is the arity of the relation $R$)
  satisfying $p_1$, \ldots, $p_n$, respectively, each of which has
  as a successor the corresponding element of the $R$-tuple. This is
  spelled out below in more detail.  The syntactic translation
  $[\cdot]^*$ from $\sunfp$ to $\ml_\mu$ is then as follows:
\begin{align*}
  &[P(x)]^* &=~ & p_P \\
  &[X(x)]^* &=~ & X \\
  &[\phi(x)\land\psi(x)]^* &=~ & [\phi(x)]^*\land [\psi(x)]^* \\
  &[\phi(x)\lor\psi(x)]^* &=~ & [\phi(x)]^*\lor [\psi(x)]^* \\
  &[\neg\phi(x)]^* &=~ & \neg[\phi(x)]^* \\
  &[[\LFP_{X,y}\phi(y)](x)]^* &=~ & \mu X ~ [\phi(y)]^* \\
  &[\exists y_1\ldots y_n(R(y_1\ldots y_n) \land y_i=x\land\bigwedge_{j\neq
    i}\phi_j(y_j))]^* &=~ &
  \Diamond^{-}(p_i\land\Diamond^{-}(p_R\land\bigwedge_{j\neq
    i}\Diamond(p_j\land\Diamond[\phi_j(y_j)]^*)))
  \\
  &[\exists x\phi(x)]^* &=~ & \jump[\phi(x)]^*
\end{align*}
For any structure $M$, we denote by $M^*$ the Kripke model obtained as follows:
every atomic fact $R(a_1\ldots a_n)$ of $M$ gets replaced by a substructure
consisting of a node labeled $p_R$ which has $n$ children, labeled $p_1, \ldots
p_n$, each of which has as its single child $a_i$. Then it is clear from the
construction that $(M,a)\models\phi(x)$ if and only if $(M^*,a)\models[\phi(x)]^*$.
Conversely, for every Kripke model $M$, we define a structure $M_*$ as follows:
the elements of $M_*$ are the elements of $M$. A fact $R(a_1, \ldots, a_n)$ is inserted in $M_*$ whenever
in $M$ there is a node labeled $p_R$ and $n$ (not necessarily distinct)
children of $p_R$, labeled $p_1, \ldots p_n$, each of which has $a_i$ as a
child.

It is clear from the construction that $(M,a)\models[\phi(x)]^*$ if and only if
$(M_*,a)\models\phi(x)$. 

Altogether this shows that the translation $[\cdot]^*$ preserves
(un)satisfiability, both on arbitrary structures and on finite structures.

(3) If $\phi$ has a finite model then by the construction above $[\phi]^*$ has
a finite Kripke model. By Theorem~\ref{thm-otto-acyclic-cover} this implies
that $[\phi]^*$ has a model that is $4l$-acyclic. Now notice that the
transformation of $M$ into $M_*$ described above preserves $l$-acyclicity up
to a factor of 4, and therefore, $\phi$ has a model that is is $l$-acyclic.

(2) is proved in the same way as (3), using the fact that if $[\phi]^*$ has a
model then it has an acyclic one and that the
transformation of $M$ into $M_*$ preserves acyclicity.

(4) For a \sunfo formula $\phi$, $[\phi]^*$ does not contain any fixpoint and
is therefore in  \ml. Recall from
Section~\ref{sec-modal-logic} that \ml has the finite model property. This
implies that \unfo also has the finite model property as the transformation of
$M$ into $M_*$ preserves finiteness.
\end{proof}

\subsection{Arbitrary \unfp-formulas}\label{sec-unfp-sat}

We now attack the satisfiability problem for arbitrary \unfp formulas.  We say
that a disjunction $\psi_1 \lor \psi_2$ is unary if $\psi_1\lor\psi_2$ has
at most one free variable.  We remind the reader that we use the notation $\psi(y)$ to express that
$\psi$ has \emph{at most} one free first-order variable $y$.

\begin{lemma}\label{lem:strongnormalform}
  Every \unfp sentence is equivalent to a \unfp sentence in
  \unnormalform that uses only unary disjunction, and, more precisely,
  a sentence generated by the following grammar:
\begin{equation}\label{simple-form}
 \chi(y) ::= \exists\tuple{z}~\psi(y,\tuple{z}) \mid \neg\chi(y) \mid
 \chi_1(y)\lor\chi_2(y) \mid [\LFP_{X,z}~\chi(z)](y)
\end{equation}
where $\psi(y,\tuple{z})$ is of the form:

  \begin{equation}\label{simple-CQ}
\exists \tuple{z}\big(\tau(\tuple{z})\land z_i=y \land
  \!\!\!\!\!\!\!\!\bigwedge_{j\in\{1\ldots n\}\setminus\{i\}}\!\!\!\!\!\!\!\!\phi_j(z_j)\big) 
\text{\qquad or \qquad}
  \exists \tuple{z}\big(\tau(\tuple{z}) \land
  \!\!\!\!\!\bigwedge_{j\in\{1\ldots
    n\}}\!\!\!\!\!\phi_j(z_j)\big)
\end{equation}
where $\tuple{z}=z_1\ldots z_n$, $i\leq n$, and $\tau(\tuple{z})$ is a
conjunction of relational atomic formulas with no equalities, and $\phi_j$ is
generated by the grammar~\eqref{simple-form}.
\end{lemma}

\begin{proof}
Concerning~\eqref{simple-form}. Clearly we may assume that $\phi$ is a
sentence as all free variables can be existentially quantified without
affecting the satisfiability and the fact that $\phi$ is in \unnormalform.
It is also easy to see that non-unary occurrences of disjunction can be
eliminated at the cost of a possible exponential blowup, using the
fact that disjunction commutes with conjunction and with the
existential quantifiers. This is an exponential transformation,
reminiscent of the transformation of propositional formula into
disjunctive 
normal form, and it does not affect the width of the formula, nor the size of 
other parameters that we will define later. This remark will be
important during the complexity analysis as the width will appear in the
exponent of the complexity of some forthcoming transformations.

The form~\eqref{simple-CQ} is straightforward to obtain by eliminating equalities by
identifying the respective quantified variables.
\end{proof}

In what follows let $\phi$ be any \unfp formula of the form described
in Lemma~\ref{lem:strongnormalform}, for which we want to test satisfiability.

We denote by $\subf_\phi$ the set of all subformulas $\psi(y)$ of $\phi$ that
have one free first-order variable.  For any subformula of $\phi$ of the form
\[\exists \tuple{z}(\tau(\tuple{z})\land z_i=y \land
  \!\!\!\!\!\!\!\!\bigwedge_{j\in\{1\ldots
    n\}\setminus\{i\}}\!\!\!\!\!\!\!\!\phi_j(z_j)) 
\text{\qquad or \qquad}
\exists \tuple{z}(\tau(\tuple{z})\land
  \!\!\!\!\!\bigwedge_{j\in\{1\ldots
    n\}}\!\!\!\!\!\phi_j(z_j)) ~,
\]
we call $\tau(\tuple{z})$ a \emph{neighborhood
  type}. We denote the set of neighborhood types in
$\phi$ by $\types_\phi$. In the sequel we will sometime view a neighborhood
type $\tau(\tuple{z})$ as a structure whose universe is the set of variables of
$\tuple{z}$ and whose relations are the atoms of $\tau$.

\newcommand{\stitch}[1]{#1^{\textsf{x}}}

We now consider structures that are ``stitched together'' from copies of the
neighborhood types in $\types_\phi$. 

\begin{definition} Consider the signature that contains an $n$-ary
  relation symbol $R_\tau$ for each neighborhood type $\tau(z_1,
  \ldots, z_n)\in \types_\phi$. A structure in this new signature is
  called a \emph{stitch diagram}.  Each stitch diagram $M$ gives rise
  to a \emph{stitching} $\stitch{M}$, which is the structure (with the
  same domain of $M$) in the original signature obtained by replacing
  each $R_\tau$-tuple with a copy of the neighborhood type $\tau$
  (viewed as a structure) for all $\tau\in\types_\phi$.
\end{definition}

At this point, our basic strategy for reducing \unfp to \sunfp should be
clear: we will produce an \sunfp sentence to describe stitch diagrams
whose stitchings satisfy the desired \unfp sentence. In the 
rest of this section, we work out the details of this strategy.

It is important to realize that, even if a stitch diagram $M$ does not
contain an atomic fact $R_\tau(\tuple{a})$, it may still be the case that
$\stitch{M}\models\tau(\tuple{a})$. In this case we say that the fact
$R_\tau(\tuple{a})$ is \emph{implicit} in $M$.  For example, this could happen
if $M\models R_{\tau'}(\tuple{a})$ and $\tau$ is contained in
$\tau'$.
The following claim gives us a handle on when this phenomenon may
occur.  For any $\tau\in\types_\phi$, we denote by $|\tau|$ the number
of atomic formulas in $\tau$. We write $N\subseteq M$ if $N$ is a
not-necessarily-induced substructure of $M$.

\begin{lemma}\label{claim-implicit}
  If $R_\tau(\tuple{a})$ is implicit in a stitch diagram $M$ then
  there is an $N\subseteq M$ containing at most
  $|\tau|$ many facts, such that $R_\tau(\tuple{a})$ is already implicit
  in $N$. Moreover $N$ is connected whenever $\tau$ is.
\end{lemma}

\begin{proof}
   We need at most one fact of $M$ to account for each atom
   in $\tau(\tuple{a})$. 
\end{proof}

Let $l = \max_{\tau\in\types_\phi} |\tau|$. We will restrict attention to
stitch diagrams $M$ that are $l$-acyclic. By Item~(3) of
Proposition~\ref{prop-sunfp} this is without loss of generality. This implies
that every $N\subseteq M$ containing at most $l$ facts is acyclic. The
importance of the above claim, then, shows in two facts: (i) intuitively, there
are finitely many reasons why a fact may be implicit in $M$, and (ii) each of
these reasons is acyclic, and hence can be described in \sunfp as we will see.

%
\begin{lemma}\label{lem:maintranslation}
Let $\psi(y,\tuple{X})$ be any subformula of $\phi$ with at most one free
first-order variable. 
By induction on the structure of $\psi(y)$ we can construct a \sunfp formula
$\psi'(y)$ such that, for all $l$-acyclic stitch diagrams $M$, all $a\in M$, and all sets $\tuple{S}$ of
elements of $M$, $M \models \psi'(a,\tuple{S})$ iff $\stitch{M}
\models \psi(a,\tuple{S})$.
\end{lemma}

\begin{proof}
The inductive translation commutes with all Boolean operators and with the \LFP~ operator.
Fix now any $\tau(\tuple{z})\in\types_\phi$ with $\tuple{z}=z_1, \ldots, z_n$, fix
an $i\leq n$, and fix a sequence of formulas $\psi_1, \ldots, \psi_{i-1},
\psi_{i+1}, \ldots, \psi_n\in\subf_\phi$ and assume $\psi$ is of the form:
\[\psi(y) ~:=~ \exists
\tuple{z}(\tau\land z_i=y\land \!\!\!\!\!\!\!\!\bigwedge_{j\in\{1\ldots
  n\}\setminus\{i\}}\!\!\!\!\!\!\!\!\psi_j(z_j))~.\]
(The argument if $\psi$ is of the form
$\exists\tuple{z}(\tau\land\bigwedge_{j\in\{i\ldots n\}}\psi_j(z_j))$ is
similar. Note that these two cases also account for the base of the
induction, if we let $n=1$).

By induction we already have constructed \sunfp formulas $\psi'_1, \ldots, \psi'_{i-1},
\psi'_{i+1}, \ldots, \psi'_n$ corresponding to $\psi_1, \ldots, \psi_{i-1},
\psi_{i+1}, \ldots, \psi_n$.

We are interested in detecting in $M$ how a node in $\stitch{M}$ may come to
satisfy $\psi$.
  We will construct a \sunfp formula that lists all the
cases in $M$ that make this happen.  It clearly suffices to consider one connected
component of $\tau$ at a time. Hence by Lemma~\ref{claim-implicit} it only
depends on a small neighborhood of $x$ in $M$.
The formula will then be essentially a long disjunction, where each disjunct
corresponds to the description of a small neighborhood of $M$ in which $\tau$ is
implicitly satisfied by a tuple of nodes satisfying in addition the formulas
$\psi_j$. Note that since we assume $M$ to be $l$-acyclic, these small
substructures are all acyclic, which will make it possible to describe them by
an (existential) formula of \sunfp.  

More precisely, consider any connected acyclic stitch diagram $N$ containing at
most $l$ facts, and any homomorphism $h:\tau\to \stitch{N}$. We now construct
an \sunfp formula $\chi_{\psi,N,h}(y)$ that describes $N$ (existentially
positively) from the point of view of $h(z_i)$, and expressing also that each
$h(z_j)$ satisfies $\psi_j$.

We shall make use of the following property of acyclic structures. If $N$ is
acyclic, i.e. $\inc(N)$ is acyclic, then there is a tree $T(N)$ such that each
of its nodes is labeled with a fact of $N$ satisfying the following properties:
(i) each atom of $N$ is the label of exactly one node of $T(N)$, (ii) if a node
$u$ of $T(N)$ is the parent in $T(N)$ of a node $v$ then their label share at
most one element of $N$, (iii) if two nodes $u$ and $v$ of $T(N)$ share an
element then either they are siblings of one is the parent of the other.
In other words, we view $N$ as a tree $T(N)$ rooted with an atom containing
$h(z_i)$ and the formula describes that tree from top to bottom.  We construct
the desired \sunfp formula by induction on the number of nodes in $T(N)$.

When $T(N)$ contains only one node, whose label is $R_\tau(a_1,\cdots,a_m)$
and, assuming $h(z_i)= a_{i_0}$, the desired formula
is then
$$\exists \tuple{y}~
R_\tau(\tuple{y})\wedge y_{i_0}=y \wedge \bigwedge_{j\neq i, h(z_j)=a_{\alpha_j}}\psi'_j(y_{\alpha_j}).
$$

Assume now that $T(N)$ is a tree whose root element is
$u=R_\tau(a_1,\cdots,a_m)$ and with several subtrees $T(N_1), T(N_2)
\cdots$. By property (ii) of $T(N)$, for all $j>0$, the label of $u$ and the
elements of $N_j$ share at most one element, say $a_{\beta_j}$. For $j> 0$, let
$h_j$ be the restriction of $h$ to the elements of $N_j$. Finally assume that
$h(z_i)=a_{i_0}$. The desired formula $\chi_{\psi,N,h}(y)$ is then:
\[
\exists \tuple{y} \Big(
R_\tau(\tuple{y}) \wedge y_{i_0}=y \wedge \!\!\!\!\!\!\bigwedge_{j\neq
  i, h(z_j)=a_{\alpha_j}}\!\!\!\!\!\!\!\!\!\!\!\!\psi'_j(y_{\alpha_j})
~~\wedge~~ \bigwedge_{j}\chi_{\psi,N_j,h_j}(y_{\beta_j}) \Big)
\]
Finally $\psi'(y)$ is the disjunction, for each $N$ and $h$ as above,
of the formulas $\chi_{\psi,N,h}(y)$.

It follows from the construction that, for all $l$-acyclic stitch diagrams $M$,
$M\models\psi'$ if and only if $\stitch{M}\models\psi$. 
\end{proof}

It follows from Lemma~\ref{lem:maintranslation} that if the
constructed \sunfo formula $\psi'$ is satisfiable (in the finite or in the
infinite), by Theorem~\ref{thm-otto-acyclic-cover} it has a $l$-acyclic model
(in the infinite this model is actually acyclic) and therefore $\psi$ is
satisfiable.
Conversely, it is easy to construct, from a (finite) model $M$ of
$\psi$, a (finite) model $M'$ of $\psi'$. Take for $M'$ the structure whose domain is
the domain of $M$ and whose relation $R_\tau$ contains all the tuples of $M$
satisfying $\tau$.
Altogether this shows that $\psi$ is satisfiable (on finite structures) if
and only if $\psi'$ is. By Proposition~\ref{prop-sunfp} this implies that \unfp is
decidable. A careful analysis of the complexity of the above translation
actually yields:

\begin{theorem}\label{theorem-sat-upper-bound}
   The satisfiability problem for \unfp is in \twoexptime, both on
   arbitrary structures and on finite structures. 
\end{theorem}

\begin{proof}
  Assume first that the input \unfp formula $\phi$ satisfies the simplifying
  assumptions of Step~1. Recall that we denote by $l$ the maximal number of
  conjuncts in neighborhood types, and by $k$ the \width of the formula.  Note
  that the formula $\chi_{\psi,N,h}(y)$ is only polynomially long in the length of
  $\psi$, but for any given $\psi$, the number of possible structures
  $N$ and homomorphisms $h$ can be
  exponential.  More precisely, each structure $N$ to be considered has at most
  $l$ facts and domain size at most $k \cdot l$. Moreover the number of
  possible atomic relations is $|\types_\phi|=O(|\phi|)$ and each relation has
  arity at most $k$ and cannot contain twice the same element. There are at
  most $|\phi|^{O(k\cdot l)}$ many such structures, up to isomorphism.  For
  each such structure $N$, since the domain size is at most $l\cdot k$, the
  number of homomorphisms $h:\tau\to\stitch{N}$ is at most $(l\cdot k)^k =
  |\phi|^{O(k)}$. All in all, the number of disjuncts occurring in $\psi'$ is
  bounded by $|\phi|^{O(k\cdot l)}$. Lets now consider the size of one such
  disjunct: as it contains formulas of the form $\psi_i'$ for smaller formulas,
  we obtain by induction that the size of the \sunfp formula $\phi'$ is bounded by
  $|\phi|^{O(r\cdot k\cdot l)}$ where $r$ is the nesting depth of existential
  blocks.

  Consider now the general case of a \unfp formula $\theta$. Putting $\theta$
  in \unnormalform is linear time. The transformation of $\theta$ into a
  formula $\phi$ satisfying the simplifying assumption is exponential in time
  but produces a formula whose parameters $k,l$ and $r$ are only polynomial
  (actually linear) in
  the size of $\theta$. Hence, from the previous paragraph, it follows that the
  size of the resulting \sunfp formula $\psi'$ can be bounded by
  $(2^{|\theta|})^{|\theta|^c}$ for some constant $c$, that is exponential in $|\theta|$.

  Hence, we obtain by Proposition~\ref{prop-sunfp} that the satisfiability
  problem is in 2-\exptime, both on arbitrary structures and on finite
  structures.
 \end{proof}

\medskip

We conclude this section by proving Theorem~\ref{thm-finite-model-prop} (Finite
Model Property of \unfo)
and Theorem~\ref{thm-tree-like-prop} (Tree-like Model Property of \unfp) using the reductions described above. 

\begin{proof}[Proof of Theorem~\ref{thm-finite-model-prop}] 
Recall the construction of
a \sunfp formula from a \unfp formula described above, and observe
that, when starting with a
formula of \unfo, we actually obtain a formula of \sunfo. Consider now a satisfiable
\unfo formula $\phi$. By construction the resulting \sunfo formula $\phi'$ is
satisfiable. By Item~(4) of Proposition~\ref{prop-sunfp} $\phi'$ has a finite
model $N$. By construction $\stitch{N}$ is a finite model of $\phi$.
\end{proof}

\begin{proof}[Proof of Theorem~\ref{thm-tree-like-prop}]
Similarly, consider a satisfiable \unfp formula $\phi$. By construction the resulting \sunfp formula $\phi'$ is
satisfiable. By Item~(2) of Proposition~\ref{prop-sunfp} $\phi'$ has an acyclic
model $N$. By construction this implies that $\stitch{N}$ is a model of $\phi$
that has tree-width at most $k-1$, where $k$ is the width of $\phi'$.
\end{proof}

\subsection{Lower bounds and restricted fragments}

The complexity result of Theorem~\ref{theorem-sat-upper-bound} is
tight:

\begin{prop}\label{prop-sat-lower-bound}
  There is a fixed finite signature such that the satisfiability problem for $\unfo$ is
  \twoexptime-hard, both on arbitrary structures and on finite structures.
\end{prop}
\begin{proof}
  Fix an alternating $2^n$-space bounded Turing machine $M$ whose word problem
  is 2-\exptime-hard. We may assume that the Turing machine runs
  in double exponential time (e.g., by maintaining a counter). 
 Let $w$ be a word in the input alphabet of $M$.  We
  construct a formula $\phi_{w}$ that is satisfiable if and only if $M$ accepts
  $w$. Moreover, if $\phi_{w}$ is satisfiable, then in fact it is satisfied in
  some finite tree structure.  In this way, we show that the lower bound holds
  not only for arbitrary structures, but also for finite trees and for any
  class in-between. The formula $\phi_{w}$ describes an (alternating) run of
  $M$ starting in the initial state with $w$ on the tape, and ending in a final
  configuration.

  The run is encoded as a finite tree whose nodes labeled by a unary
  predicate $C$ represent configurations of the Turing machine, and
  where a child-edge between two $C$-nodes indicates that the second
  node represents a successor configuration of the first. Each
  $C$-node is also labeled by a unary predicate $Q_i$ indicating the state
  of the Turing machine in that configuration. In addition,
  each $C$-node is the root of a subtree uniformly of height $n$, in which
  each non-leaf node has a child satisfying $P_0$ and a child
  satisfying $P_1$ (and no children satisfying both). We can associate to every leaf node of this
  subtree a number between $0$ and $2^n$, determined by whether each of
  its $n$ ancestors (the node itself included) satisfies $P_0$ or $P_1$. Thus, each leaf node of the
  subtree represents a tape cell, and, using further unary predicates,
  we can encode at each tape cell the current content of that tape
  cell, and whether the head is
  currently located there. See Figure~\ref{fig:TMencoding}. All in
  all, the schema of the structure consists of a binary relation $R$,
  unary relations $C, P_0, P_1$, a unary relation $Q_i$ for each 
  state of the Turing machine, a unary relation for each element of
  the alphabet, and a unary relation $H$ to represent the head
  position.

\begin{figure}
\begin{center}
\includegraphics[scale=0.35]{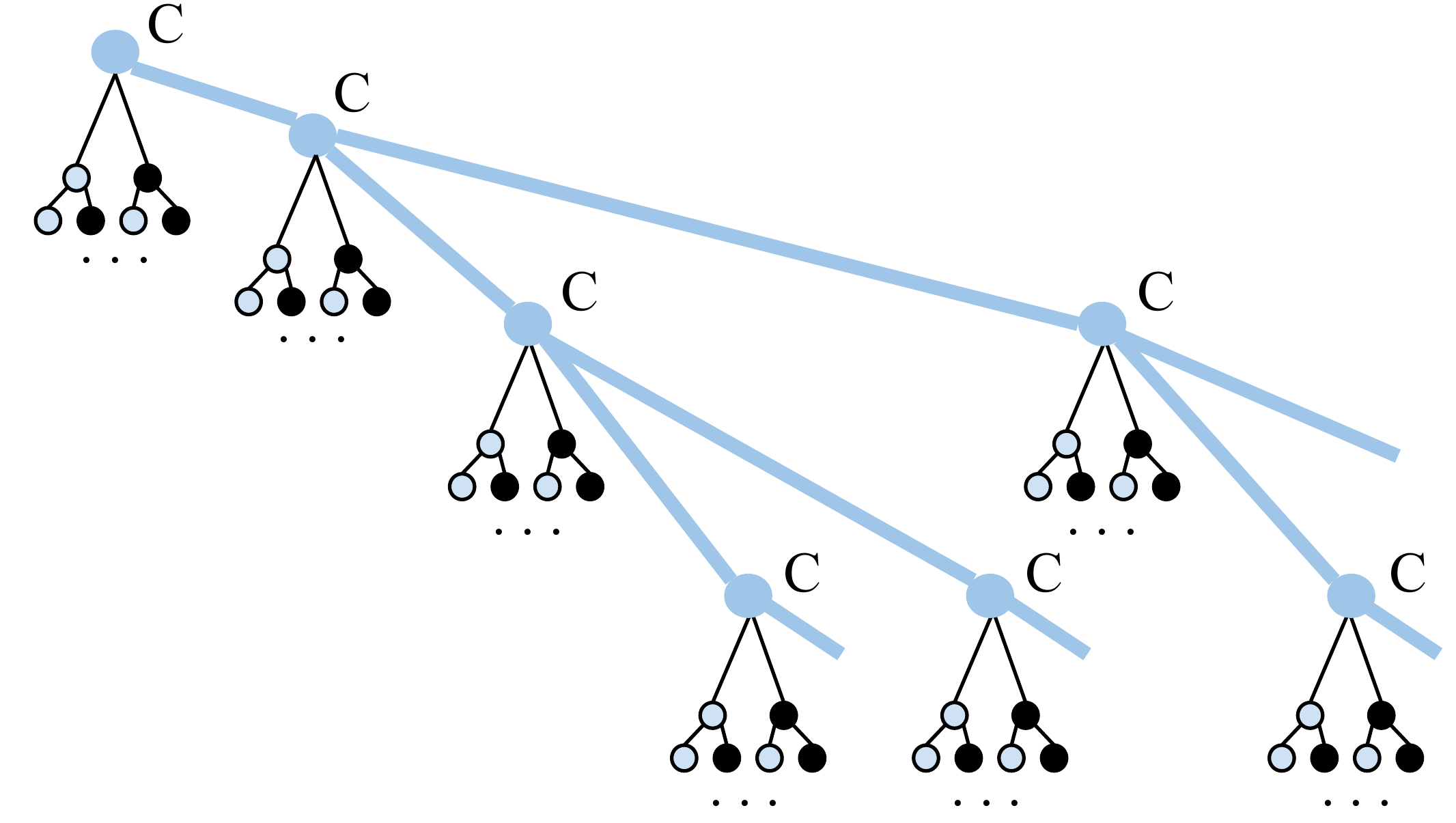}
\end{center}
\caption{Encoding of a run of the alternating Turing machine}\label{fig:TMencoding}
\end{figure}

The construction
  of the formula $\phi_w$ is based on the above encoding of runs as
  structures. More precisely, $\phi_w$ is the conjunction of
\begin{itemize}
\item A formula $\phi_1$ expressing that the root node is labeled $C$, 
    and that
    each $C$-node is the root of a subtree uniformly of height $n$ in which every non-leaf node has a child satisfying
    $P_0$ and a child satisfying $P_1$ (and no children satisfying both).
\item A formula $\phi_2$ expressing that, whenever two leaf nodes
    represent the same tape cell in the same configuration, then they 
    agree on all unary predicates (note that, in \unfo, we cannot 
    force that there is only one node representing each tape cell in a 
    given configuration, because we cannot express
    inequality).
\item A formula $\phi_3$ expressing that, whenever two leaf nodes
    represent the same tape cell in the successor configuration, and 
    this tape cell is not the head position, then the two nodes
    agree on all unary predicates. 
\item A formula $\phi_4$ encoding the transition function of the 
   Turing machine (i.e., whenever, in some configuration, the Turing
  machine is in a $\exists$-state ($\forall$-state), and its head reading a particular
  letter, then for some (respectively, for every) possible transition
  there exists a corresponding successor configuration).
\item A formula $\phi_5$ expressing that, in the initial (root)
  configuration, the tape content is $w$ and the Turing machine is
  in the initial head position and state; and all final configurations
  (i.e., $C$-nodes without $C$-successors) are accepting configurations.
\end{itemize}

\noindent We omit a detailed definition of the formulas in question, which is
tedious but not difficult. For example, if we use
$x\uparrow^n\downarrow^m y$ as a shorthand for a \unfo formula stating
that there is a path going $n$ steps up in a tree from $x$ and then
$m$ steps down reaching $y$, if we use $\text{leaf}(x)$ as a shorthand
for  $\neg\exists y R(x,y)$, and if we use $P_1^{\uparrow i}(x)$ as a
shorthand for a \unfo formula that $x$ has an $i$ step ancestor satisfying
$P_1$, then  $\phi_2$ can be expressed
  as follows (for every unary predicate $A$):
\begin{equation*}
\neg\exists x y(\text{leaf}(x)\land \text{leaf}(y)\land
x\uparrow^n\downarrow^n y \land \bigwedge_{i=0\ldots n-1} (P_1^{\uparrow i}(x) \leftrightarrow
P_1^{\uparrow i}(y)) \land A(x) \land\neg A(y))
\end{equation*}
For $\phi_3$ and $\phi_4$ we make use of the following \unfo formula expressing
the fact that $x$ and $y$ denote the same tape position in successive
configurations:
\begin{equation*}
\text{leaf}(x)\land \text{leaf}(y)\land (x\uparrow^{n+1}\downarrow^{n+2}y)\land
\bigwedge_i(P_1^{\uparrow i}(x)\leftrightarrow P_1^{\uparrow i}(y))
\end{equation*}
The rest of the construction is straightforward.
\end{proof}

The above proof actually established the lower bound for arbitrary structures,
for finite trees, and on any class in-between. Moreover, the proof only uses
formulas that have negation depth~2. For formulas of negation
depth~1, the satisfiability problem turns out to have a lower
complexity.

\begin{thm}\label{thm:negationdepth1}
  The satisfiability problem for $\unfo$ formulas of negation depth 1 is
  \npnp-complete (even for formulas containing unary predicates only).
\end{thm}

The proof of Theorem~\ref{thm:negationdepth1} is postponed until
Section~\ref{section-model-check} as it makes use of results
regarding the complexity of model checking which are obtained there.

\section{Model Checking}\label{section-model-check}
In this section we study the complexity of the model-checking problem for \unfo
and \unfp. The model-checking problem takes as input a structure $M$ and a
sentence $\phi$, and it asks whether $\phi$ is true in $M$. We focus here on
the \emph{combined complexity} of the model-checking problem, where the input
consists of a sentence and a structure.  It was already observed in
\cite{tencate07:navigational} that the model checking problem for \unfo (there
called CRA(mon$\neg$)) is in \pnp . Here, we show that the problem is in fact
$\pnplogsq$-complete, and that the model checking problem for $\unfp$ is in
$\npnp\cap\conpnp$.  We refer to the reader to Section~\ref{section-complexity}
for the definition of the relevant complexity classes.

\subsection{Model checking for \unfo}

We start by showing that the model checking problem for \unfo is
{\pnplogsq-complete}. We then show that bounding the nesting of negations
lowers the complexity to \pnplog-complete and that allowing subformula-sharing increases the
complexity to \pnp-complete.

In order to give a better intuition about the upper bound algorithms we start
by presenting a naive evaluation algorithm. Recall that the \emph{negation
  depth} of a \unfo formula is the maximal nesting depth of negations in its
syntactic tree.  Let $\phi(x)$ be a formula of \unfo of negation depth~$l$ and $M$
a model. Testing for an element $u \in \dom(M)$, whether $M\models
\phi(u)$, can be done by induction on~$l$ using the following bottom-up strategy.

If $l=0$ then $\phi$ is existential and we can test whether $\phi$ holds using
one call to an \np-oracle (the oracle guesses witnesses fror the existentially quantified
variables and then checks whether the remaining quantifier-free part
of the formula holds).

If $l\neq 0$, then, for all elements $v$ of $M$, and all maximal subformulas $\psi(y)$ of
$\phi$ that have negation depth $l-1$, we test whether $M\models \psi(v)$
using the induction on $l$. All these tests can be performed
independently of each other, and hence (by induction hypothesis) we
need only $l-1$  parallel calls to the \np-oracle. Based on the
outcomes of these calls, we
compute a new structure $M'$ by expanding $M$ with new unary predicates
recording for each element of $M$ which of the maximal subformulas
$\psi(y)$ of negation depth $l-1$
hold at that element. We also transform $\phi$ into $\phi'$ replacing
$\neg\psi(y)$ with the appropriate newly introduced unary predicate. 
It remains to evaluate $\phi'$ on $M'$ using the same algorithm as in  the base case where $l=0$.

Altogether this yields an algorithm that makes $l$ parallel calls
to an \np-oracle. When $l$ is constant this implies that the total process can
be made in \pnplog and this gives the upper bound of
Theorem~\ref{thm:bounded-negation} below. When $l$ is at most $\log |\phi|$
this implies that the total process can be made in \pnplogsq (recall the
discussion in Section~\ref{section-complexity}) and
the upper bound of Theorem~\ref{thm-unfo-model-check} essentially reduces
the general case to this case.

\begin{theorem}\label{thm-unfo-model-check}
 The model checking problem for \unfo is $\pnplogsq$-complete.
\end{theorem}
\begin{proof}
    For the \textbf{lower bound}, we give a 
    reduction from LEX$_2$(SAT): 
\emph{given a Boolean formula $\phi(x_1....x_n)$, is $x_{\lceil\log^2(n)\rceil}$ true
  in the lexicographically maximal satisfying assignment?}

Let $\phi(x_1, \ldots, x_n)$ be a given Boolean formula, and let $d=\log(n)$
(we may assume that $n$ is a power of 2). We are interested in knowing whether
$x_{d^2}$ is set to $1$ in the lexicographically maximal satisfying assignment
of $\phi$. For this we construct, in time polynomial in $n$, a model $M_n$ and
a formula $\theta_n$ such that $M_n \models \theta_n$ iff $x_{d^2}$ is set to
$1$ in the lexicographically maximal satisfying assignment of $\phi$.

Let $M_n$ be the structure containing $n$ elements, $a_1, \ldots, a_n$,
together with two elements, $1, 0$. The elements $1,0$ represents truth and
falsity and are distinguished from the elements $a_1, \ldots, a_n$ using a unary
predicate $Q$ that holds only for $0,1$.  Each of the elements $a_1, \ldots,
a_n$ of $M_n$ represents a bit-string of length $d$ encoded using $d$ unary
predicates, $P_1,\cdots,P_d$, with the intended meaning that $P_i(a_j)$ holds
in $M_n$ iff the $i$-th bit of the bit-string represented by $a_j$ is true.
$M_n$ is such that $a_1, \cdots, a_n$ code all possible bit-strings of length
$d$.

Below, as a suggestive notation, we will write $x_i$ for variables intended to
range over truth and falsity, and $y_i$ for variables intended to range over
the $n$ elements of $M_n$ that represent length-$d$ bit-strings. Hence $\exists
x~ \psi$ should be understood as $\exists x~ Q(x) \land \psi$ and $\exists y~
\psi$ as $\exists y~ \lnot Q(y) \land \psi$.

By $\widehat{\phi}$ we denote the \unfo formula obtained from $\phi$ by replacing,
for $k\leq d$ and $j\leq d$, the variable $x_{(k-1)d+j}$ by
$P_j(y_k)$. 
Note that the free variables of $\widehat{\phi}$ are $y_1,\ldots,y_d$ and
$x_{d^2+1}, \ldots, x_n$. As $y_1,\ldots,y_d$ will range over bit strings of
length $d$, we are interested in the last bit of $y_d$, i.e. $P_d(y_d)$.

We define, by induction on $i$, a formula $\chi_i(y)$ that is true for an
element $a_j$ if the length-$d$ bit-string represented by $a_j$ describes the
bits $x_{(i-1)d+1}\cdots x_{(i-1)d+d}$ of the lexicographically maximal
satisfying assignment of $\phi$. It is convenient to define simultaneously
another formula, $\psi_i(y)$, which is true for an element $a_j$ if the
length-$d$ bit-string represented by $a_j$ describes the bits
$x_{(j-1)d+1}\cdots x_{(j-1)d+d}$ in the \emph{some} satisfying assignment
whose prefix up to $x_{(j-1)d}$ is the same as the lexicographically maximal
satisfying assignment of $\phi$.
\begin{eqnarray*}
  \psi_i(y) &=& \exists y_1\ldots y_d ~~x_{d^2+1} \ldots
x_n~~~\big(\widehat{\phi} \land y_i=y \land \bigwedge_{j<i}\chi_j(y_j)\big) \\
  \chi_i(y) &=& \psi_i(y) \land \neg\exists y'(y'>y \land \psi_i(y'))
\end{eqnarray*}
where $y'>y$ is shorthand for a formula expressing that the bit-string denoted
by $y'$ is lexicographically greater than the bit-string denoted by $y$:
\begin{equation*}
\bigvee_{i<d} \big((\bigwedge_{j<i} P_j(y) \leftrightarrow P_j(y'))  \wedge
\lnot P_i(y) \wedge P_i(y') \big).
\end{equation*}
\smallskip

\noindent Finally, take 
\begin{equation*}
\theta_n = \exists y_1\ldots y_d~~ x_{d^2+1} \ldots
x_n~~~(\widehat{\phi}\land\bigwedge_{i\leq d}(\chi_i(y_i))\land P_d(y_d)).
\end{equation*}
Then $\theta_n$ is true in $M_n$ if and only if $x_{d^2}$ is true in the
lexicographically maximal satisfying assignment of $\phi$ as required.
Moreover $\theta_n$ is indeed a formula of \unfo.
Finally notice that each $\chi_i$ has a size exponential in $i$ but they are used only for $i\leq d=\log
n$. Hence $\theta_n$ has polynomial size and can be computed in time polynomial
in $n$.

\smallskip

We now turn to the \textbf{upper bound}. Recall the algorithm presented in the
preamble of this section. If the syntactic tree of the formula would be a
balanced binary tree then its negation depth $l$ would be bounded by the $\log$
of the size of the formula and we would be done. The idea is to evaluate the
formula by first making its syntactic tree balanced. This trick has already
been applied in the context of branching time model checking using ``Tree Block
Satisfaction''~\cite{Schnoebelen03:oracle}. Instead of redoing the trick, we
reduce the model checking problem of \unfo to this one.

We first present ``Tree Block Satisfaction'', called TB(SAT)$_{1\times M}$
in~\cite{Schnoebelen03:oracle}. In fact, for present purposes, it suffices to
consider a restricted version of the problem TB(SAT)$_{1\times M}$, which we
will refer to simply as TB(SAT) in what follows.  We will describe here this
restricted version.  A TB-tree of width $k\geq 1$ is a tree consisting of
blocks, where each block is, intuitively, a kind of Boolean circuit having $k$
output gates and having $k$ input gates for each of its children. See
Figure~\ref{fig:tbtrees}. The $i^{th}$ output of a block is defined in terms of
the input gates by means of an existentially quantified Boolean formula
$\chi_i$ of the form
 \[\exists \tuple{b}_1 c_1 \ldots
\tuple{b}_m c_m \tuple{d}~~~\big(c_1=\mathsf{input}_{i_1}(\tuple{b}_1)\land\cdots\land
c_m=\mathsf{input}_{i_m}(\tuple{b}_m)) \land
\psi\big)\]
where each $\tuple{b}_j$ is a tuple of $\log k$ Boolean variables, 
encoding a number $[\![\tuple{b}_j ]\!]$ from $1$ to $k$, and
$\mathsf{input}_{i_j}(\tuple{b}_j)$ represents the value of the
$[\![\tuple{b}_j]\!]$-th output bit of the $i_j$-th child block
(which is denoted by $y^{(i_j)}_{[\![\tuple{b}_j]\!]}$ in Figure~\ref{fig:tbtrees})
and $\psi$ is a Boolean formula using any of the existentially
quantified Boolean variables.  
\begin{figure}
\begin{center}
\includegraphics[scale=0.4]{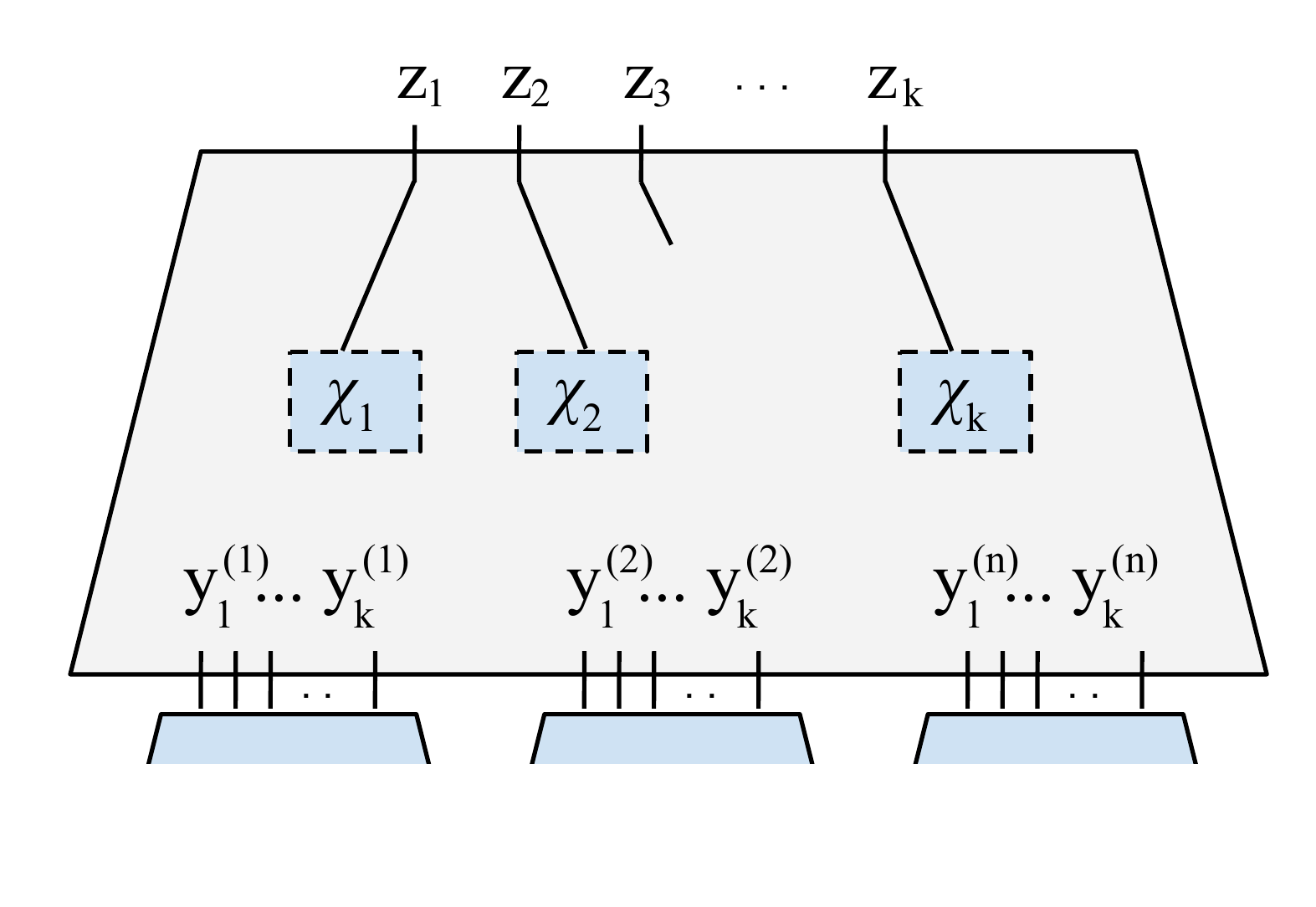}
\end{center}
\vspace{-5mm}
\caption{A block in a TB-tree with $n$ children}\label{fig:tbtrees}
\end{figure}

TB(SAT) is then the problem, given a TB-tree and a truth assignment
for all inputs of the blocks that are leaves of the TB-tree, whether
the first output bit of the root block evaluates to true. It was shown
in~\cite[corollary 3.4]{Schnoebelen03:oracle} that TB(SAT) can be
decided in \pnplogsq.

Given a formula $\phi$ of \unfo and a structure $N$, we construct in polynomial
time a TB-tree that is a yes instance of TB(SAT) iff $N\models\phi$. The
reduction is simple and reflects the naive evaluation of $\phi$ on $N$.

The construction of the TB-tree is by induction on $\phi$. The width $k$ of the
TB-tree is the size of the domain of $N$. At each step of the induction the
block of the root of the TB-tree is associated to a subformula of $\phi$ having
one free variable, and is such that its output gate $i$ is set to $1$ iff the
$i$-th element of $N$ makes the associated subformula true. Altogether the
shape of the TB-tree resemble the one of the syntactic tree of the formula.
In the QBF formulas constructed below we denote by $\tuple{b}$ a vector of
$\log k$ variables and $[\![\tuple{b}]\!]=i$ is a shorthand for the Boolean
formula stating that $\tuple{b}$ represents the binary encoding of $i$.

\medskip

Let $\phi(x)$ be any formula in one free variable.

\medskip

Case 1: $\phi(x)$ is of the form $\neg\psi(x)$. By induction hypothesis, we
have a TB-tree for $\psi(x)$ such that the outputs from the root block indicate
which elements of $N$ satisfy $\psi(x)$. We extend this TB-tree with one extra
block on top, in which the $i$-th output gate is defined by the formula
\[\exists
\tuple{b},c~~~(c=\mathsf{input}(\tuple{b}) \land [\![\tuple{b}]\!]=i \land c=0).
\]
 This formula sets its output
to $1$ iff its input $i$ is set to $0$, hence preserves the inductive
hypothesis as required. 

\medskip

Case 2: $\phi(x)$ is built up from atomic formulas and formulas in one free
variable using conjunction and disjunction and existential quantification. Let
$y_1, \ldots, y_n$ be the set of all existentially quantified variables in
$\phi$. Let $\psi_1(z_1)\ldots, \psi_m(z_m)$ be the maximal subformulas in
one free variable occurring in $\phi$, where $z_1 \ldots z_m$ are among
$x,y_1, \ldots, y_n$. We may assume without loss of generality that $z_1,
\ldots, z_m$ are the first $m$ variables from the sequence $y_1, \ldots, y_n$
(in particular that $z_1, \ldots, z_m$ are distinct variables) because if not,
one can always introduce another existentially quantifier variable $y_{n+1}$
and replace $\psi_j(z_j)$ by $(\psi_j(y_{n+1})\land y_{n+1}=z_j)$.  By
induction, we have a TB-block $T_j$ for each $\psi_j(y_j)$ such that the
outputs from the root block of $T_j$ indicate which elements of $N$ satisfy
$\psi_j(y_j)$.  The TB-tree for $\phi$ will consist of a new root block
whose children are the roots of each of the $T_j$. The definition of the $i$-th
output gate is
  \[\exists
  \tuple{b}_1,c_1,\ldots,\tuple{b}_m,c_m,\tuple{b}_{m+1},\ldots,\tuple{b}_n~~~\big
  (\bigwedge_{j\leq
    m}~~~(c_j=\mathsf{input}_j(\tuple{b}_j))\land \chi_N\big)\]
 where $\chi_N$ is obtained from
  $\phi$ by replacing each subformula of the form $\psi_l(y_l)$ by $c_l$, each subformula
  of the form $x=y_l$ by $\tuple{b}_l=i$, each subformula of the form $y_l = y _{l'}$ by
  $\tuple{b}_l=\tuple{b}_{l'}$ and each subformula of the form $R(y_{l_1},\ldots,y_{l_\kappa})$ by a 
  Boolean formula listing all tuples in the relation $R^N$:
\[
\bigvee_{([\![\tuple{d}_1]\!],\ldots,[\![\tuple{d}_\kappa]\!])\in R^N} (\tuple{b}_{l_1}=\tuple{d}_1 \land
\ldots \land  \tuple{b}_{l_\kappa}=\tuple{d}_\kappa)
\]

\medskip

Note that Case~2 covers the base case when $\phi(x)$ has no subformula
with one free variable.

It is now easy to check that the TB-tree constructed by the above induction has
the desired property and can be computed in polynomial time. This concludes the
proof of the Theorem.
\end{proof}

As mentioned earlier, restricting the nesting of negations gives a lower complexity.

\begin{theorem}\label{thm:bounded-negation}
For all $l>0$, the complexity of the model-checking problem for \unfo formula
of negation depth bounded by $l$ is \pnplog-complete. The lower bound holds
even 
for a fixed structure. 
\end{theorem}
\begin{proof}
  For the {\bf lower bound} we use the equivalent characterization of \pnplog
  as the class of problems that are \ptime truth-table reducible to NP. As
  3-colorability is \np-complete, every problem in \pnplog is \ptime
  truth-table reducible to 3-colorability. Recall from
  Section~\ref{section-complexity} that a \ptime truth-table reduction from a
  given problem to 3-colorability is a \ptime algorithm that, given an instance
  of the problem, produces a set $y_1,\cdots,y_n$ of inputs to 3-colorability,
  together with a Boolean formula $\phi(x_1,\cdots,x_n)$, such that the input
  is a yes-instance iff $\phi$ evaluates to true after replacing each $x_i$ by
  $1$ if $y_i$ is 3-colorable and $0$ otherwise.

  We show that we can reduce any problem that is \ptime truth-table reducible
  to 3-colorability to the model checking problem for \unfo formula of negation
  depth 1. For this, fix a problem $P$ that is \ptime truth-table reducible to
  3-colorability. Let $w$ be an input string for $P$.  We construct from $w$,
  in time polynomial in $|w|$, an instance $M$ and a formula $\phi_w$ such that
  $M \models \phi_w$ iff $w \in P$. As $M$ will be independent of $w$ and
  $\phi_w$ will be in \unfo and with nesting depth 1, this will show
  \pnplog-hardness.

 By assumption, there is a polynomial time algorithm that produces from $w$ a
 tuple $\langle G_1,\cdots,G_n, C(x_1, \ldots, x_n)\rangle$ such that $w\in P$
 iff $C(x_1, \ldots, x_n)$ evaluates to 1 after replacing each of the $x_i$ by
 1 iff $G_i$ is 3 colorable.

 For each graph $G$, let $q_G$ be the \emph{canonical conjunctive query of
   $G$}, that is, the existentially quantified conjunction of relational atoms,
 where there is an existentially quantified variable for each node of $G$, and
 a relational atom for each edge of $G$. Note that, since $q_G$ does not use
 negation, it belongs to \unfo. Let $M$ be the structure representing a clique
 of size 3.  It is well known that $M\models q_G$ iff $G$ is 3-colorable.  Let
 $\phi_w$ be the formula constructed from $C$ by replacing each occurrence
 of $x_i$ with the sentence $q_{G_i}$.  It is immediate to verify that the
 resulting formula is in \unfo, has negation depth 1 and satisfies the
 properties required for the reduction. Moreover $\phi_w$ can be computed in
 time polynomial in $|w|$ as this was the case for $\langle G_1,\cdots,G_n,
 C(x_1, \ldots, x_n)\rangle$. This completes the proof for the lower-bound.

\medskip

The {\bf upper bound} is obtained using the naive evaluation of the formula
presented in the preamble of this section as $l$ is treated as a constant.
\end{proof}

We now consider formalism equivalent in expressive power to \unfo but with a
more succinct syntax that allows the sharing of subformulas.  Let us denote by \unfolet the
extension of \unfo with Boolean variables $b_1, b_2, \ldots$ (ranging over
truth values) that can be used as atomic formulas, and with a new construct
$\mathsf{let} ~b =\phi ~\mathsf{in}~ \psi$, where $b$ is a Boolean variable,
$\phi$ is a sentence, i.e., a formula without free first-order variables but
possibly with free Boolean variables (excluding $b$ itself), and $\psi$ is a formula
that may use all the Boolean variables, including $b$. We only
consider
\unfolet formulas without free Boolean variables, i.e., in which each occurring
Boolean variable is bound by a $\mathsf{let}$ operator. The semantics of
\unfolet formulas is as expected: when the
valuation of the Boolean free variables of $\phi$ is known, we can derive a
valuation for $b$ using $b=\phi$ and then we can evaluate $\psi$. 

\begin{example}
A typical formula of \unfolet looks like this:
\[\mathsf{let} ~ b = \psi ~ \mathsf{in} ~ \exists x_1,\ldots x_n~~~
  \big(b=0 \land \psi_1(\tuple{x})\big)
 \lor
 \big(b=1 \land \psi_2(\tuple{x})\big)
\]
It is equivalent to the \unfo formula:
\[ \exists x_1,\ldots x_n~~~
 \big(\lnot \psi \land \psi_1(\tuple{x})\big)
 \lor
 \big(\psi \land \psi_2(\tuple{x})\big)
\]
\end{example}

The \unfolet
model checking problem is the problem of evaluating a given \unfolet formula
without free Boolean variables in a given structure.

\begin{theorem}\label{thm:let}
  The \unfolet model-checking problem is \pnp-complete. The lower bound holds
even  for a fixed structure. 
\end{theorem}
\begin{proof}
  For the {\bf upper bound}, we use a simple bottom-up evaluation
  strategy. Let $\phi$ be any \unfolet formula. We may assume
  without loss of generality that each let-operator in $\phi$ binds a 
  different Boolean variable. We assign a rank to each Boolean
  variable
  occurring in $\phi$: a Boolean variable $b$ has rank 0 if its
  definition (i.e., the sentence to which it is bound by its
  let-operator
  in $\phi$) does not contain any Boolean variables  (free or
  bound),
   and $b$ has rank $k+1$ if its definition only contains Boolean 
  variables of rank at most $k$. 
  Using Theorem~\ref{thm-unfo-model-check} repeatedly as an oracle, we can compute 
  in polynomial time a truth value for each Boolean variable. 
  Finally, by replacing all Boolean variables in $\phi$ by their
  truth value and applying Theorem~\ref{thm-unfo-model-check} once more, 
  we find out whether $\phi$ is true in the structure.
  Altogether this yields a \pnp algorithm.

  For the {\bf lower bound}, we make use of the problem \lexsat: \emph{given a
    satisfiable Boolean formula $\phi(x_1, \ldots, x_n)$, test if the value of
    $x_n$ is 1 in the lexicographically maximal solution}
  (cf.~Section~\ref{section-complexity}).  
  Given a satisfiable Boolean formula $\phi(x_1, \ldots, x_n)$ we construct in
  time polynomial in $n$ a model $M$ and a \unfolet formula $\psi$ such that $M
  \models \psi$ iff $x_n$ is 1 in the lexicographically maximal solution of
  $\phi$.

  The idea of the reduction will be to construct the lexicographically maximal
  solution $B_1\cdots B_n$ of $\phi$ bit by bit from $B_1$ to $B_n$ using the
  following algorithm: $B_1$ is true iff $\phi(\top,x_1,\ldots,x_n)$ is
  satisfiable and if $B_1\cdots B_i$ have already been computed then $B_{i+1}$
  is true iff $\phi(B_1,\cdots,B_i,\top,x_{i+1},\cdots,x_n)$ is satisfiable.

  Specifically, our reduction uses a fixed structure $M$ with two elements, one
  of which is labeled by a unary predicate $T$ and intuitively represents
  \emph{true} while the other element represents \emph{falsity}.
  We let $\widehat{\phi}$ be the \unfo formula obtained from $\phi$ by
  replacing each positive occurrence of the variable $x_i$ in $\phi$ by the
  formula $T(x_i)$ and each occurrence of $\lnot x_i$ in $\phi$ by the formula $\lnot T(x_i)$.

  We construct by induction sentences $\psi_1,\ldots, \psi_n$, where each
  $\psi_i$ is true in $M$ if and only if $B_i$ is 1.  By definition $B_1=1$ iff
  $\phi$ has an satisfiable assignment setting $x_1$ to true. This can be
  expressed using the formula:
  \[\psi_1 := \exists x_1\ldots x_n ~~~ (\widehat{\phi}(x_1,\ldots.
  x_n)\land T(x_1))~,\] 
By definition the following \unfo formula, obtained by induction, expresses that $B_i$ is 1: 
\[
\psi_i := \exists x_1\ldots x_n~~~ (\widehat{\phi}(x_1,\ldots, x_n)\land
  (\bigwedge_{j<i} T(x_j)\leftrightarrow \psi_j)\land
  T(x_i)).
\]
Notice that the size of $\psi_i$ is exponential in $i$ and therefore, even
though $\psi_n$ does express the truth value of $B_n$, this would not give a
polynomial time reduction. However, using the $\mathsf{let}$ construction of
\unfolet we can derive the same formulas more succinctly.

Let $\psi'_i$ be the formula (with $b_1,\ldots,b_{i-1}$ as free Boolean variables):
\[
\exists x_1,\ldots x_n
  (\widehat{\phi}(x_1, \ldots, x_n)\land (\bigwedge_{j<i} T(x_j) \leftrightarrow
  b_j)\land T(x_i))~.
\]

Then we set $\psi$ be the \unfolet formula:
\begin{eqnarray*}
\psi := \mathsf{let} & b_1 = \psi'_1 &  \mathsf{ in }\\ 
 &\mathsf{let } & b_2 = \psi'_2 \quad \mathsf{ in }\\
 && \ldots\\
 & & \mathsf{let } \quad b_{n-1} = \psi'_{n-1} \quad \mathsf{ in }  \quad 
\psi'_n
\end{eqnarray*}

Then notice that $\psi$ has all the desired properties: it is in \unfolet, it
can be computed in time polynomial in $n$ and it verifies that $M\models \psi$
iff $B_n=1$. This concludes the proof.
\end{proof}

\subsubsection*{Satisfiability for $\unfo$ formulas of negation depth 1}

We are now also in a position to give the proof of
Theorem~\ref{thm:negationdepth1},
which states that testing satisfiability for \unfo 
formulas of negation depth 1 is \npnp-complete.

\begin{proof}[Proof of Theorem~\ref{thm:negationdepth1}]
  Let $\psi$ be any \unfo formula of negation depth 1, and let $\chi$
  be the first-order formula obtained by bringing $\phi$ into prenex
  normal form. Note that $\chi$ is, in general, no longer a \unfo
  formula. Since $\psi$ does not contain nested negations or universal
  quantifiers, $\chi$ is of the form
  $\exists\tuple{x}\forall\tuple{y}\chi'$, where $\chi'$ is a
  quantifier-free formula, which we may assume to be in
  negation-normal form.  In addition, we know that the variables
  $\tuple{y}$ can only occur in negative atomic subformulas of
  $\chi$. In other words, all positive atomic subformulas of $\chi$
  use only variables in $\tuple{x}$.  We claim that $\chi$, and hence,
  $\psi$, has a ``polynomial size model property'': suppose $M\models
  \forall\tuple{y}\phi(\tuple{a},\tuple{y})$. Let $M'$ be the submodel
  of $M$ containing only the elements $\tuple{a}$ and containing only
  the (polynomially many) facts about $\tuple{a}$ that occur
  positively in $\phi(\tuple{a},\tuple{y})$. It is easy to see that
  $M'$ still satisfies
  $\forall\tuple{y}\phi(\tuple{a},\tuple{y})$. The existence of a
  model of $\phi$ can now be tested by guessing a structure of polynomial
  size and then applying the model checking procedure to verify whether
  it is a model or not. Using Theorem~\ref{thm:bounded-negation}, this
  implies that the satisfiability problem for \unfo-formulas of
  negation depth 1 is in \npnp.

  Hardness is by reduction from the problem of evaluating QBF formulas of the form
       $$\phi := \exists\tuple{x} \forall\tuple{y} \psi$$
   with $\tuple{x}=x_1\ldots x_n$ and $\tuple{y}=y_1\ldots y_m$ which is known
   to be complete for \npnp.
   We will construct in \ptime a \unfo formula of negation depth 1 that is
   satisfiable if and only if $\phi$ evaluates to true.
   Take the vocabulary consisting of unary predicates $T$ and $F$ plus unary
   predicates $P_1, \ldots, P_n$ corresponding to the existentially quantified variable $\tuple{x}$.
   The \unfo formula is defined as the conjunction 
\begin{align*}
&\exists x.T(x) \land \exists x.F(x) \land  \neg\exists x.(T(x) \land  F(x))\\
\wedge~~ &\bigwedge_{i\leq n}\exists x(P_i(x) \wedge  (T(x) \vee F(x))) \land \neg(\exists x.(P_i(x) \land
T(x)) \land \exists x(P_i(x) \land F(x))) \\
\wedge~~ &\neg\exists\tuple{y} (\bigwedge_i (T(y_i) \vee F(y_i)) \land -\phi)
\end{align*}
where $-\phi$ is obtained by negating $\phi$, then pushing the
negations down to the atoms, and then
 replacing $y_i$ by $T(y_i)$,
 replacing $\neg y_i$ by $F(y_i)$,
 replacing $x_i$ by $\exists x(P_i(x) \wedge T(x))$,
and replacing $\neg x_i$ by $\exists x(P_i(x) \wedge F(x))$.

It is straightforward to verify that this \unfo formula is satisfiable if and
only if $\phi$ evaluates to true, and that it has negation depth 1.
\end{proof}

\subsection{Model checking for \unfp}

We now turn to the complexity of the model checking of \unfp.  It is
convenient, at this point, to treat the greatest-fixpoint operator (GFP) as a
primitive operator, instead of as a defined connective. This way we can use
dualization in order to assume without loss of generality that all formulas of
\unfp are such that \emph{every occurrence of a fixpoint operator in $\phi$
  is a positive occurrence (i.e., lies under an even number of negations).}
We also assume without loss of generality that each fixpoint operator binds a
different variable, so that we can speak of \emph{the fixpoint definition} of a
variable $X$, by which we mean the formula directly below the fixpoint operator
binding $X$.  The \emph{dependency graph} of fixpoint variables in $\phi$ is
the directed graph whose nodes are the fixpoint variables in $\phi$ and
where there is an edge from $X$ to $Y$ if $Y$ occurs (free or bound) inside the
fixpoint definition of $X$.  We say that a \unfp formula is \emph{alternation
  free} if the dependency graph does not have a cycle containing both a least
fixpoint variable and a greatest fixpoint variable \cite{Niwinski1986}.

We first consider the alternation-free fragment.  

\begin{theorem}\label{thm-modelcheck-unfp-altfree}
The model-checking problem for the alternation-free fragment of \unfp
is $\pnp$-complete.
\end{theorem}

\begin{proof}
  Recall that we assume the formulas are such that \emph{every occurrence of a
    fixpoint operator in $\phi$ is a positive occurrence (i.e., lies under
    an even number of negations).}

  We first prove the claim for formulas containing only $\LFP$ operators (no
  $\GFP$ operators).  Let $\phi$ be any \unfp formula containing only least
  fixpoint operators and $M$ be a model.  Let $\tuple{X}=X_1,\ldots, X_n$ be the fixpoint
  variables occurring in $\phi$ (we assume each fixpoint operator binds a
  different variable).  Initially, we assign each $X_i$ to be the empty set. We
  then repeatedly consider each fixpoint variable $X_i$ and evaluate its
  fixpoint definition $\beta_i$, viewed as a \unfo formula by using the current
  choice of sets $\tuple{X}$ to interpret the free set variables of $\beta_i$
  as well as the fixpoint subformulas of $\beta_i$, and check if new elements
  are derived that do not already belong to the set $X_i$. If this is the case,
  we add the elements in question to $X_i$.  We repeat this procedure until no
  new elements are derived.  It is well known that the resulting sets we end up
  with are the least fixpoint solutions for the fixpoint subformulas of
  $\phi$. Furthermore, the number of iterations is polynomial, since, in
  each iteration, at least one element of $M$ gets added to on of the sets, and
  each iteration can be performed in \pnplogsq by Theorem~\ref{thm-unfo-model-check}.

  By dualization, we get the same result for formulas containing only greatest
  fixpoint operators.  The result is then easily lifted to the full alternation
  free fragment by induction on the alternation rank of
  the fixpoint variable in question, where the alternation rank is defined as
  the (finite) maximal number of fixpoint alternations on an outgoing path from
  that variable in the dependency graph: we simply perform a bottom up evaluation based on
  the syntactic tree defining $\phi$, coloring each node of $M$ with the
  fixpoint formulas that it satisfies using the above algorithms.

  For the lower bound, we reduce from \lexsat. Let $\phi(x_1, \ldots, x_n)$ be
  any Boolean formula that is the input of the \lexsat problem. We may assume
  without loss of generality that $\phi$ is in negation normal form.  We
  construct a structure $M$, whose domain is $\{t_1, \ldots, t_n, f_1, \ldots,
  f_n\}$, and with unary predicates $T$ and $V_1, \ldots, V_n$, such that
  $T^M=\{t_1, \ldots, t_n\}$, and $V_i^M=\{t_i,f_i\}$. Intuitively, each
  element $t_i$ represents the eventuality that the value of $x_i$ in the
  lexicographically maximal satisfying assignment is $1$, while $f_i$
  represents the eventuality that the value of $x_i$ in the lexicographically
  maximal satisfying assignment is $0$. Using a fixpoint formula, we will
  compute the set of all elements that are actually ``true'' in the
  lexicographically maximal satisfying assignment. It then suffices only to
  check whether $t_n$ belongs to this set.

  Let $\phi'$ be obtained from $\phi$ by replacing every occurrence of
  a Boolean variable $x_i$ by $T(x_i)$. 
  Let $\theta$ be the formula $\exists z(V_n(z)\land T(z)\land [\LFP_{X,x}\psi](z))$
  where $\psi$ is the disjunction of all formulas of the following
  forms, for all $1\leq k\leq n$:
  \[V_k(x)\land T(x) \land \exists x_1, \ldots,
  x_n(\bigwedge_{i=1\ldots n}\!\!\!\! V_i(x_i)\land \phi' \land
  X(x_1)\land \cdots \land X(x_{k-1})\land T(x_k))\]
and
  \[V_k(x)\land F(x) \land \neg\exists x_1, \ldots,
  x_n(\bigwedge_{i=1\ldots n}\!\!\!\! V_i(x_i)\land \phi' \land
  X(x_1)\land \cdots \land X(x_{k-1})\land T(x_k))\]
  It is easy to see from the construction that $\theta$ is true in $M$
  if and only if $x_n=1$ in the lexicographically maximal satisfying
  assignment for $\phi$. 
\end{proof}

We now turn to the general case.

\begin{theorem}\label{thm-modelcheck-unfp}
The \unfp model checking problem is in $\npnp\cap \conpnp$ and  \pnp-hard.
\end{theorem}
\begin{proof}

The {lower bound} is immediate from
Theorem~\ref{thm-modelcheck-unfp-altfree}.
We prove here the {upper bound}. 

Since \unfp is closed under unary negation,
it is enough to show that the model-checking problem is in \npnp.  We show a
slightly stronger result proving that, for every formula in at most one free
variable, we can decide in \npnp whether a given element of a given structure 
makes the formula true.  The algorithm we describe below is inspired by an
idea from \cite{Vardi95:bounded} to reduce the problem to the case of formulas
that only contain least fixpoint operators (and no greatest fixpoint operators)
by guessing a set for each greatest fixpoint operator.
 
In the sequel we will only consider formulas $\phi$ with one free
first-order variable and by ``evaluating'' this formula over $M$, 
intuitively, we mean computing, non-deterministically, 
a set of elements of $M$ making the formula true. 

Recall from the beginning of the section that we consider
\unfo-formula in which we allow both LFP and GFP-operators, and 
we require that all occurrences of fixpoint operators are positive.
Equivalently, we view each \unfp formula as being defined inductively by the grammar:
\begin{align*}
\phi &::= \alpha(\tuple X,\tuple \psi,x)\\
\psi &::=  [\LFP_{Y,y}~\phi(Y,\tuple{X},y)](x) ~|~ [\GFP_{Y,y}~\phi(Y,\tuple{X},y)](x) 
\end{align*}
where $\alpha$ is a \unfo formula with one free first-order variable, possibly
several monadic second-order variables, and possibly using as atoms fixpoint subformulas
$\psi(z)$ defined by mutual induction using the above grammar, such
that $\tuple X$ and $\tuple \psi$ occur only positively (under an even
number of negations) in $\alpha$.

Given a formula of the form $\alpha(\tuple X,\tuple \psi,x)$ we denote by $\hat
\alpha$ the \unfo formula constructed from $\alpha$ by replacing each nested
fixpoint subformula $\psi(z)$ by $Y(z)$ where $Y$ is the variable defined by
$\psi$. Hence $\hat\alpha$ has $x$ as unique first-order free variable and
$\tuple X$ and $\tuple Y$ as  monadic second-order free variables. The positivity
condition stated above implies that $\hat \alpha$ is monotonic with respect to
$\tuple X$ and $\tuple Y$.

\newcommand\eval{\textit{eval}\xspace}

Given a structure $M$, sets $\tuple{U}$ of elements of $M$ and a \unfo formula
$\phi(\tuple X,x)$ we denote by $\eval(\phi,(M,\tuple{U}))$ the subset of
the domain of $M$ computed by induction on the syntactic representation of
$\phi$ using the following algorithm where, as a convenient notation, we
also write $\eval(\tuple\phi,(M,\tuple{U}))$ to denote the tuple
$(\eval(\phi_1,(M,\tuple{U})),\ldots,\eval(\phi_n,(M,\tuple{U})))$, where
$\tuple\phi=\phi_1,\ldots,\phi_n$:

\bigskip\noindent
\emph{$\eval(\phi,(M,\tuple{U}))$}
\begin{itemize}
\item Case $\phi$ is of the form $\alpha(\tuple X,\tuple \psi,x)$
  \begin{enumerate}
  \item Compute $\tuple V:=\eval(\tuple \psi, (M,\tuple U))$ by induction
  \item Evaluate $\hat \alpha$ on $(M,\tuple U,\tuple V)$, using the
    algorithm of Theorem~\ref{thm-unfo-model-check}
  \end{enumerate}

\item Case $\phi$ is of the form $[\LFP_{Y,y}~\alpha(Y,\tuple{X},\tuple\psi,y)](x)$
  \begin{enumerate}
  \item set $S:=\emptyset$
  \item\label{step-lfp} Compute $\tuple V:=\eval(\tuple \psi,
    (M,\tuple U,S))$ by induction
  \item Compute $S':=\eval(\hat\alpha, (M,\tuple U,\tuple V,S))$ by induction
  \item If $S'=S$ return $S$ otherwise go to Step~\ref{step-lfp} with $S:=S\cup
    S'$
  \end{enumerate}
\item Case $\phi$ is of the form $[\GFP_{Y,y}~\alpha(Y,\tuple{X},\tuple \psi,y)](x)$
  \begin{enumerate}
  \item guess $T$
  \item Compute $\tuple V:=\eval(\tuple \psi, (M,\tuple U,T))$ by induction
  \item Compute $T':=\eval(\hat\alpha, (M,\tuple U,\tuple V,T))$ by induction
  \item If $T'=T$ return $T$ otherwise abort
  \end{enumerate}
\end{itemize}

\noindent It should be clear that the algorithm is in \npnp: by monotonicity
it makes only polynomially many inductive calls, and the base case corresponds to
the evaluation of a \unfo formula and is therefore in \pnplogsq by
Theorem~\ref{thm-unfo-model-check}.

It remains to show that this result returns the correct answers.
Recall the semantics of fixpoints as described in Section~\ref{section-prelim}.
It is easy to see that the algorithm below computes exactly the correct results
according to this semantics.

\bigskip\noindent
\emph{$\eval^*(\phi,(M,\tuple{U}))$}
\begin{itemize}
\item Case $\phi$ is of the form $\alpha(\tuple X,\tuple \psi,x)$
  \begin{enumerate}
  \item Compute $\tuple V:=\eval(\tuple \psi, (M,\tuple U))$ by induction
  \item Evaluate $\hat \alpha$ on $(M,\tuple U,\tuple V)$
  \end{enumerate}

\item Case $\phi$ is of the form $[\LFP_{Y,y}~\alpha(Y,\tuple{X},\tuple\psi,y)](x)$
  \begin{enumerate}
  \item set $S:=\emptyset$
  \item\label{step-lfp1} Compute $\tuple V:=\eval(\tuple \psi,
    (M,\tuple U,S))$ by induction
  \item Compute $S':=\eval(\hat\alpha, (M,\tuple U,\tuple V,S))$ by induction
  \item If $S'=S$ return $S$ otherwise go to Step~\ref{step-lfp1} with $S:=S\cup
    S'$
  \end{enumerate}
\item Case $\phi$ is of the form $[\GFP_{Y,y}~\alpha(Y,\tuple{X},\tuple \psi,y)](x)$
  \begin{enumerate}
  \item set $T=\dom(M)$
  \item\label{step-gfp} Compute $\tuple V:=\eval(\tuple \psi, (M,\tuple U,T))$ by induction
  \item Compute $T':=\eval(\hat\alpha, (M,\tuple U,\tuple V,T))$ by induction
  \item If $T'=T$ return $T$ otherwise go to Step~\ref{step-gfp} with $T:=T
    \cap T'$
  \end{enumerate}
\end{itemize}
Let's denote by $\eval^*(\phi,(M,\tuple{U}))$ the set computed by this second
algorithm, i.e. the standard algorithm for fixpoint formulas.  Notice that it
differs with the previous algorithm only in the case of greatest fixpoints. The
fact that $\eval(\phi,(M,\tuple{U}))=\eval^*(\phi,(M,\tuple{U}))$ is a
consequence of the following claim.

\begin{claim}
For all $M,\tuple U$ and $\phi$ we have $\eval(\phi,(M,\tuple{U})) \subseteq
\eval^*(\phi,(M,\tuple{U}))$. Moreover, if \eval always guess the
correct greatest fixpoint then we have quality.
\end{claim}
\begin{proof}
  This is a simple induction on $\phi$. For Case~(1) we obtain by induction
  that $\tuple V \subseteq \tuple V^*$ and the result follows immediately by the
  monotonicity assumption on $\hat \alpha$. Similarly, for Case~(2), we obtain
  that each stage of the computation of the least fixpoint by \eval is included
  in the same stage by $\eval^*$. The inclusion of the respective least
  fixpoints follows. For Case~(3) we obtain again by induction that $\tuple V
  \subseteq \tuple V^*$ and therefore that $T \subseteq \eval(\hat\alpha, (M,\tuple
  U,\tuple V^*,T))$ by the monotonicity assumption on $\hat \alpha$. If we denote
  by $T^*$ the greatest fixpoint of $\hat\alpha$ on $(M,\tuple U,\tuple V^*)$, by
  Knaster-Tarski Theorem, this implies that $T \subseteq T^*$ as desired.

  The second part of the claim is immediate.
\end{proof}
This concludes the proof of Theorem~\ref{thm-modelcheck-unfp}.
\end{proof}

\section{Trees}\label{section-tree}

In this section, we study the expressive power and computational
complexity of \unfo and \unfp on trees. 
We consider three types of trees: \emph{binary trees} (with two deterministic
successor relations $child_1, child_2$, and any number of unary predicates), \emph{unranked
  trees} (with a single child relation, and any number of unary predicates), and \emph{XML
  trees} (that is, trees in which nodes can have any number of
children and the children of each node are ordered, and where the
signature consists of the horizontal and vertical successor and order
relations, and any number of unary
predicates).  We will consider finite trees, but all proofs generalize to
infinite trees.

On XML trees with all axes, it is known that \emph{Core XPath = $\fo^2$} for
unary queries while \emph{Core XPath = $UCQ$-over-$\fo^2$-unary-predicates} for
binary queries \cite{Marx05:semantic}. It turns out that \unfo has the same
expressive power as Core XPath, both for unary and for binary queries
(cf.~\cite{tencate07:navigational}, where \unfo is called CRA(mon$\neg$))
and therefore \unfo characterizes Core XPath in a
more uniform way.  In particular, since the XML tree languages definable in
Core XPath are precisely the ones definable in $\fo^2$~\cite{Marx05:semantic},
this implies that $\unfo$ defines the same XML tree languages as $\fo^2$. The
same holds already for $\unfo^2$.

In this section, we further analyze the expressive power and complexity of
\unfo and \unfp on trees.

The following observation will be helpful. We say that two unranked
trees are \emph{root-to-root bisimilar} if the roots of the two trees
are bisimilar.

\begin{lemma}\label{lem:bisimulation-trees}
  Two unranked trees are UN-bisimilar if and only if they are root-to-root bisimilar.
\end{lemma}

\begin{proof}
  It is clear that every UN-bisimulation is a root-to-root
  bisimulation. Conversely suppose $t, t'$ are root-to-root bisimilar
  trees, with roots $r$ and $r'$.  For any node $a$ of $t$, we denote
  by $depth_t(a)$ the distance from $r$ to $a$, and we denote by by
  $t_a$ the subtree of $t$ rooted at $a$. Similar notations apply to
  $t'$. Let $Z$ consist of all pairs $(a,b)$ of nodes from $t$ and
  $t'$, respectively, such that (i) $depth_t(a)=depth_{t'}(b) = k\geq
  0$, and (ii) for each $i\leq k$, the subtree of $t$ rooted at the
  $i$-th ancestor of $a$ and the subtree of $t'$ rooted at the $i$-th
  ancestor of $b$ are root-to-root bisimilar. 

  We claim that $Z$ is a UN-bisimulation.  Let $(a,b)\in Z$. We show how
  to construct a homomorphism $h$ from $t$ to $t'$ that maps $a$ to
  $b$. The other direction is established in the same way.

  The homomorphism $h$ is constructed as follows: first, the
  root-to-root bisimulation between $t_a$ and $t'_b$ induces a
  homomorphism from $t_a$ to $t'_b$. If $a=r$ and $b=r'$, we are
  done. Otherwise, let $a'$ be the parent of $a$ and let $b'$ be the
  parent of $b$. Then the root-to-root bisimulation between $t_{a'}$
  and $t'_{b'}$ induces a homomorphism from $t_{a'}$ to $t'_{b'}$,
  which, we may assume, extends the previously constructed homomorphism
  from $t_a$ to $t'_b$. Repeating the same argument, after $k$ many
  steps, we obtain a homomorphism from $t$ to $t'$ that maps $a$ to
  $b$, and we are done.
\end{proof}

We start with \unfp. Recall that \unfp is included into \mso. Hence, over all
kind of trees, \unfp sentences only define regular languages (that is,
\mso-definable classes of trees). The converse is
also true:

\begin{thm} The following hold both over the class of finite trees and
  over the class of finite and infinite trees: 
\begin{enumerate}
\item On binary trees, $\unfp$ defines the regular languages.
\item On XML trees,  $\unfp$ defines the regular languages. 
\item On unranked trees, $\unfp$ defines the root-to-root bisimulation invariant regular
  languages.
\end{enumerate}
  The same hold for $\unfp^2$.
\end{thm}

\begin{proof}
  It is known that all the three statements hold for the
  \mucalc~\cite{jani:expr96} (as well as for monadic
  Datalog~\cite{GottlobKoch:monadic}). 
  Hence, as the \mucalc is a fragment of \unfp
  this shows the first two claims and one direction of the third claim.
  Moreover, as every formula of the \mucalc is equivalent to a $\unfp^2$
  sentence, \unfp collapses to $\unfp^2$ over trees.
  Clearly, \unfp can only define regular languages that are invariant under
  UN-bisimulation. Hence the other direction of the third claim follows from
  Lemma~\ref{lem:bisimulation-trees}.
\end{proof}

We now turn to \unfo. The $k$-neighborhood of a node of a tree is the subtree
rooted at that node, up to depth $k$. A binary tree language is LT
(\emph{Locally Testable}) if membership into this language is determined by the
presence or absence of isomorphism-types of $k$-neighborhoods for some
$k$. Similarly, an unranked tree language is ILT (\emph{Idempotent Locally
  Testable}) if membership is determined by the presence of absence of
bisimulation-types of $k$-neighborhoods, for some $k$.

\begin{thm} The following hold both over the class of finite trees and
  over the class of finite and infinite trees:
\begin{enumerate}
\item On  binary trees,
$\unfo$ defines the LT regular languages. 
\item On  unranked  trees,
$\unfo$ defines the ILT regular languages.
\end{enumerate}
  The same hold for $\unfo^2$.
\end{thm}

\begin{proof}
  Since binary trees have bounded degree, there are only finitely many
  isomorphism types of $k$-neighborhoods for any given $k$.  Moreover, each can
  be completely described by a \unfo formula.  It follows that the LT regular
  languages can be defined in $\unfo$ (in fact, in $\unfo^2$). Incidentally,
  note that the only negation used in this construction is Boolean
  negation (i.e., negation applied to sentences),
  except for expressing the fact that a node is the root or is a leaf.

  For unranked trees, the ILT regular languages are precisely the ones that can
  be defined by a global \ml formula as defined in
  Section~\ref{sec-modal-logic}~\cite{Place09:locally}. It is clear that this
  language is contained in \unfo, and therefore all ILT regular language are
  definable by a \unfo-sentence (in fact, a $\unfo^2$-sentence).

  For the other direction, let $\phi$ be a sentence of \unfo. Without loss of generality we can assume
  that $\phi$ is in \unnormalform and satisfies the simplifying assumptions
  described in Step~1 of Section~\ref{sec-unfp-sat}. 

  We can further assume without loss of generality that the conjuncts
  $\tau(\tuple{z})$ occurring in the formulas:
  \begin{equation*}
\exists \tuple{z}\big(\tau(\tuple{z})\land z_i=y \land
  \!\!\!\!\!\!\!\!\bigwedge_{j\in\{1\ldots n\}\setminus\{i\}}\!\!\!\!\!\!\!\!\phi_j(z_j)\big) 
\text{\qquad or \qquad}
  \exists \tuple{z}\big(\tau(\tuple{z}) \land
  \!\!\!\!\!\bigwedge_{j\in\{1\ldots
    n\}}\!\!\!\!\!\phi_j(z_j)\big)
\end{equation*}
are connected when seen as structures. If this were not the case, let $I$ be
the set of indices $j$ such that $z_j$ is in the component of $y$,
$\tuple{z_I}$  and $\tau_I$ the corresponding fragments of $\tuple{z}$ and
$\tau$. Let $J$, $\tuple{z_J}$ and $\tau_J$ be parts containing the remaining indices. Then
$$\exists \tuple{z}\big(\tau(\tuple{z})\land z_i=y \land
  \!\!\!\!\!\!\!\!\bigwedge_{j\in\{1\ldots
    n\}\setminus\{i\}}\!\!\!\!\!\!\!\!\phi_j(z_j)\big)$$
 is equivalent to 
$$\exists \tuple{z_I}\big(\tau_I(\tuple{z_I})\land z_i=y \land
\bigwedge_{j\in I}\phi_j(z_j)\big) ~~~~\wedge~~~~ \exists
\tuple{z_J}\big(\tau_J(\tuple{z_J}) \land \bigwedge_{j\in J} \phi_j(z_j)\big)$$
Notice that the right-hand side of the resulting formula is a
sentence. Therefore $\phi$ is equivalent to the disjunction of $\phi_1$ and
$\phi_2$ where $\phi_1$ is the conjunction of $\exists
\tuple{z_J}\big(\tau_J(\tuple{z_J}) \land \bigwedge_{j\in J} \phi_j(z_j)\big)$
with $\phi$ where this right-hand side part was replaced with \emph{true}, and $\phi_2$ is the
conjunction of $\lnot \exists \tuple{z_J}\big(\tau_J(\tuple{z_J}) \land
\bigwedge_{j\in J} \phi_j(z_j)\big)$ and $\phi$ where this right-hand side part was replaced by
\emph{false}.

In summary, we can assume that $\phi$ is a Boolean combination of sentences in
\unnormalform, satisfying the simplifying assumption, starting with an
existential quantifier, and such that all its neighborhood types are connected.
Hence it is enough to consider a single such sentence.

Let $x$ be the first existentially quantified variable of $\phi$, i.e. $\phi$
is $\exists x \psi(x)$. By our assumptions on $\phi$, all quantified
variables of $\psi$ can be taken to 
range over the neighborhood of $x$ up to distance $|\phi|$. Hence, whether
$\psi(x)$ holds or not at a node $a$ of a tree $T$ only depends on the
neighborhood of $a$ up to distance $|\phi|$.

In the binary tree case, there are only finitely many such neighborhoods and
each of them is implied by the existence of the isomorphism-type of a
$k$-neighborhood for $k=2|\phi|$. Hence $\phi$ describes a language in LT.

In the unranked tree case, there are infinitely many such
neighborhoods. But, as \unfo is invariant under UN-bisimulation, it is
enough to consider those neighborhoods up to UN-bisimulation. Each of
the UN-bisimulation classes of these neighborhoods is implied by the
existence of the UN-bisimulation-type of a $k$-neighborhood for
$k=2|\phi|$. By Lemma~\ref{lem:bisimulation-trees} (and the
characterization of ILT on unranked trees in terms of global ML),
this
implies that $\phi$ describes a language in ILT.
\end{proof}

We conclude this section by investigating the complexity of satisfiability.

\begin{thm}
  The satisfiability problem for \unfo and for \unfp is
  2ExpTime-complete on binary trees, on ranked trees, and on XML
  trees.
\end{thm}

\begin{proof}
  For the lower bound, recall that the proof of
  Proposition~\ref{prop-sat-lower-bound} was based on an encoding of Turing
  machine runs as finite trees. Hence it applies here too.

  For the upper bound, we will consider the case of \unfp on XML
  trees, as all other cases can be seen as a special case of this one.
  We will
  describe an exponential-time translation to $\mu$Regular XPath,
  the extension of Core XPath \cite{GottlobKP05} with the Kleene star
  and with the least fixpoint operator, for which
  satisfiability on XML trees can  decided in
  ExpTime~\cite{Balder-regular-xpath}.
  We briefly recall the syntax of $\mu$Regular XPath
  (cf.~\cite{Balder-regular-xpath} for more details). 
  The language has two sort of expressions, 
  \emph{node expressions} $\phi$ and \emph{path expressions} $\alpha$
  which are defined by mutual recursion:

  \[\alpha ~~::=~~ \uparrow ~\mid~ \downarrow~\mid~ \leftarrow~\mid~
  \rightarrow ~\mid~ ~.~ ~\mid~ \alpha[\phi] ~\mid~ \alpha/\beta ~\mid~ \alpha\cup\beta~\mid~
  \alpha^*\]
  \smallskip
  \[\phi ~~::=~~ P_i ~\mid~ \top ~\mid~ \neg\phi ~\mid~ \phi\land\phi ~\mid~
  \langle\alpha\rangle ~\mid~ X ~\mid~ \mu X\phi\]
  \smallskip

  \noindent where, in node expressions of the form $\mu X\phi$, the variable $X$
  is required to occur only positively (i.e., under an even number of
  negations) in $\phi$. 

  Let $\phi$ be any \unfp-sentence, or, more generally, a
  \unfp-formula with at most one free variable.  We may assume without
  loss of generality that $\phi$ is in \unnormalform. In other words,
  $\phi$ is built up from atomic formulas using (i) negation, (ii)
  fixpoint operators, and (iii) existential positive formulas.  The
  translation to $\mu$Regular XPath is now by induction.  The base
  case of the induction, as well as the induction step for negation
  and for the fixpoint operators, is immediate.  Now, suppose that
  $\phi$ is given by an existential positive formula built from atomic
  relations and subformulas in one free variable for which we already
  know that there exists a translation. Here, we can apply the known
  result from \cite{gott:conj04,Benedikt05:structural} that, on XML
  trees, every positive existential first-order formula in one free
  variable (over the given signature) can be translated to a Regular
  XPath expression (in fact, to a Core XPath expression) in
  exponential time.  Note that when we apply this translation to a
  \unfp formula, we may treat subformulas in one free variable as
  unary predicates).
\end{proof}

\section{Discussion}\label{section-discussion}

\subsection{Logics that are contained in UNFO and UNFP}
\label{sec:logics-contained}
We have seen that unary negation logics  generalize UCQ, \ml, monadic Datalog and
\mucalc. We list here other related formalisms.

\subsubsection{Unary conjunctive view logic}
First-order unary-conjunctive-view logic (UCV) was introduced in
\cite{Bailey10:logical} as a fragment of \fo.  A UCV query is an equality-free
first-order formula over a signature consisting of unary predicates only, but
where each of these unary predicates is in fact a view defined by a unary
conjunctive query. It is easy to see that every 
UCV query can be expressed by a \unfo-formula. 
Indeed, a quantifier elimination argument (cf.~\cite{hodg:mode93})
shows that, over signatures consisting only of monadic relations,
every equality-free first-order formula is equivalent to a \unfo
formula). \unfo can be viewed as a generalization of UCV
where views may be defined in terms of each other (but without cyclic
dependencies between view predicates). 

\subsubsection{The temporal logic \ctlx}
\ctlx is the fragment of the temporal logic $\mathrm{CTL}^*$ in which only the modal operator
$\mathsf{X}$ (``next'') is allowed. More precisely, the syntax of \ctlx can be
defined as follows:
\[ \begin{array}{llll}
  \text{State formulas:} & \phi &::=& p \mid \mathsf{E}\alpha 
    \mid \mathsf{A}\alpha \mid \phi\land\phi \mid \phi\lor\phi \mid  \neg\phi \\
   \text{Path formulas:} & \alpha &::=& \phi \mid \mathsf{X}\alpha
   \mid \alpha\land\alpha \mid \alpha\lor\alpha \mid \neg\alpha
 \end{array}\]
 The semantics is the usual one for $\mathrm{CTL}^*$, and we refer the interested reader
 to~\cite{Clarke99:model}. One can show that there is a polynomial truth-preserving 
 translation from \ctlx formulas to \unfo formulas.
 The model checking problem for \ctlx is known to be complete for the complexity
 class \pnplogsq~\cite{Schnoebelen03:oracle}. This can be used to provide an alternative
 proof of the \pnplogsq-hardness of the model checking problem of \unfo, cf. Theorem~\ref{thm-unfo-model-check}.

\subsubsection{Description logics} 
The basic
description logic $\mathcal{ALC}$
(which is a notational variant of the basic multi-modal logic \textbf{K})
can be viewed as a fragment of \unfo. The same holds for a number of
extensions of $\mathcal{ALC}$. Moreover, the 
\emph{query answering problem} for these description 
logics reduces to the entailment problem for \unfo. Recall that
the query answering problem
is the following problem
(cf.~\cite{DL-handbook} for basic terminology): 
\begin{itemize}\item[]
Given a TBox (i.e., set of concept inclusions) $T$, an ABox (set
  of atomic formulas
  speaking about individuals $c_1\ldots c_n$) $A$, a
  conjunctive query $q(x_1,\ldots,x_k)$, and a tuple of individuals
  $(c_{i_1}, \ldots, c_{i_k})$, is $(c_{i_1}, \ldots, c_{i_k})$ an
  answer to $q$ in every model of $T\cup A$?
\end{itemize} 
It is easy to see that this is equivalent to the validity of the \unfo-entailment

 \[ \phi_T\land \bigwedge \!A[c_1/x_1, \ldots, c_n/x_n] ~~\models~~ q(x_{i_1}, \ldots, x_{i_k}) \]
 where $\phi_T$ is the \unfo-translation of $T$.
 This  (together with Remark~\ref{remark-entailment}) gives a new proof of the known result that 
 query answering for $\mathcal{ALC}$ is decidable and that it has the 
 finite model property (cf.~\cite{Lutz08}).
 Note that the same argument works not only for $\mathcal{ALC}$ but for any description logic 
 whose TBoxes can be expressed in \unfp. Moreover, the argument works not only
 for conjunctive queries, but for any class of queries expressible in \unfp.

\subsection{Comparison with guarded logics}
We have already seen in Example~\ref{example-prelim} that \unfo and \gfo are
incomparable in terms of expressive power: It is easy to show that the \gfo formula
$\forall xy (R(x,y)\to S(x,y))$ is not invariant under UN-bisimulations and
therefore not expressible in \unfo. On the other hand a simple argument shows
that the \unfo formula $\exists yzu (R(x,y)\land R(y,z)\land
  R(z,u)\land R(u,x))$ is not invariant under guarded-bisimulations and
  therefore not expressible in \gfo.

  The decidability and expressibility results obtained in this paper for \unfo
  and \unfp have many similarities with those of their modal logic
  counterparts. Actually several proofs are reductions to the modal
  counterparts. This is in contrast with guarded logics that often require
  new and difficult arguments. A typical example is the finite model property of \gfo whose proof
  is based on the Herwig Extension Theorem (cf.~\cite{Gradel99}). In
  contrast, we prove the finite model property for \unfo by reduction
  to the analogous result for \ml (which has a very simple proof using
  filtration, cf.~\cite{BlackburndRV02}).

  It is possible to reconcile the unary negation approach and the guarded
  approach into one logical formalism called \emph{guarded negation logic},
  where it is possible to negate a formula if all its free variables are
  guarded: $R(\tuple{x})\land \neg\phi(\tuple{x})$. The first-order and
  fixpoint formalisms obtained this way generalize both the unary negation and
  guarded approaches and enjoy all the nice properties of \unfo and
  \unfp~\cite{BCS11}.

\subsection{Undecidable extensions}
Our results show that \unfo and \unfp are well behaved logics. One may ask if
there are extensions that are still well behaved. 
Inequalities are a minimal form of negation not supported by \unfo.
Unfortunately, extending \unfo with inequalities leads to
undecidability. 
Let us denote by $\unfo^{\neq}$ the extension of \unfo with
inequalities, and by $\unfo^\neg$ the extension of \unfo 
with negative relational atomic formulas.
Recall that a fragment of first-order logic is called a
\emph{conservative reduction class} if there is a computable map from 
arbitrary first-order formulas to formulas in the fragment, which
preserves (un)satisfiability as well as finite (un)satisfiability.

\begin{theorem}\label{thm-extension-unfo}
  $\unfo^{\neq}$ and $\unfo^\neg$ are conservative reduction classes,
  and hence undecidable for satisfiability on arbitrary structures
  and on finite structures.
\end{theorem}

\begin{proof}
  It is known \cite{Borger97:classical} that $\fo^1$ with two unary
  functions and equality is a conservative reduction class. We can
  translate this logic into $\unfo^{\neq}$ with two binary relations (we can
  use inequalities to express that the two relations are graphs of
  functions, as in $\neg\exists xyz(Rxy\land Rxz\land y\neq z)$,
  and atomic formulas such as $f(x)=g(x)$ are expressed by
  $\exists yz(F(x,y)\land G(x,z)\land y=z)$). This shows that
  $\unfo^{\neq}$ is a conservative reduction class. 

  For $\unfo^\neg$, we use a similar argument. Let $E, F, G$ be binary
  relations. Using negative atomic formulas, it is possible to express
  that $E$ is an equivalence relation, as in $\neg\exists xyz(Exy\land
  Eyz\land \neg Exz)\land \neg\exists xy(Exy\land\neg Eyx)\land
  \neg\exists x\neg Exx$, and that $F, G$ are graphs of functions
  defined on equivalence classes of $E$, as in 
  $\neg\exists x\neg\exists y(Fxy)\land \neg\exists
  xx'yy'(Exx'\land Fxy\land
  Fx'y'\land \neg Eyy')$, and similarly for $G$. We then use the same
  reduction as in the case of $\unfo^{\neq}$, except that $E$
  takes the role of the equality relation.
\end{proof}

Also, in the fixed point case, one can wonder whether the restriction to
\emph{monadic} least fixed-points was necessary. Indeed, this question
naturally arises since it is known that the \emph{guarded fragment} of
first-order logic is decidable even when extended with (guarded) fixed point
operators of arbitrary arity. However, adding non-monadic fixpoint
operators to our setting make the logic undecidable.
\begin{theorem}\label{thm-extension-unfp}
  The extension of \unfp with non-monadic fixed point operators is
  undecidable for satisfiability on arbitrary structures and
  on finite structures.
\end{theorem}

\begin{proof}
  It was shown in \cite{Shmueli93:equivalence} that the containment problem
  for Datalog is undecidable on finite structures and on arbitrary structures
  (the result is only stated in~\cite{Shmueli93:equivalence} for finite
  structures, but the proof applies to any class of structures that contains
  all encodings of finite strings, under some encoding).

  We reduce this problem to the problem at hand.
  A containment between two Datalog queries $\Pi_1,
  \Pi_2$ holds precisely if
  $\exists\tuple{x}(\phi_{\Pi_1}(\tuple{x})\land\bigwedge_i P_i(x_i))
  \land \neg\exists
  \tuple{x}(\phi_{\Pi_2}(\tuple{x})\land\bigwedge_i P_i(x_i))$ is unsatisfiable,
  where $\phi_{\Pi_i}$ is the Datalog query $\Pi_i$ written as a
  formula of \fp. Note that $\phi_{\Pi_i}$ does not contain negation,
  and therefore belongs to the extension of \unfp with non-monadic
  fixed point operators. Incidentally, note that the overall
  reduction uses only Boolean negation (i.e., negation applied to sentences).
\end{proof}

\paragraph{\bf Acknowledgment} We are grateful to Philippe Schnoebelen for useful discussions
  and pointers to relevant literature. We also thank the reviewers for
  their helpful and extensive comments.

\bibliographystyle{plain}
\bibliography{unneg}

\end{document}